\documentclass[11pt]{article}
\usepackage[utf8]{inputenc}
\pdfoutput=1

\RequirePackage{amsthm,amsmath,amsfonts,amssymb}
\RequirePackage[authoryear]{natbib}%% uncomment this for author-year citations
\RequirePackage[colorlinks,citecolor=blue,urlcolor=blue]{hyperref}
\RequirePackage{graphicx}
\usepackage{bm}
\usepackage{mathtools}
\usepackage{bigstrut}
\usepackage{amsmath}
\usepackage{amssymb}
\usepackage{amsfonts}
\usepackage{multirow}
\usepackage{amsthm}
\usepackage{enumitem}
\usepackage{caption,authblk,geometry,setspace}
\theoremstyle{plain}

\newtheorem{theorem}{Theorem}[section]
\newtheorem{lemma}[theorem]{Lemma}

\newtheorem{proposition}{Proposition}
\theoremstyle{definition}

\newtheorem{example}{Example}

\theoremstyle{remark}

\DeclareMathOperator*{\argmax}{arg\,max}

\newcommand{\diff}{\mathop{}\!\mathrm{d}}
\newcommand{\cvcs}{\text{CV}_{\text{cs}}}
\newcommand{\cvws}{\text{CV}_{\text{ws}}}
\newcommand{\bzero}{\mathbf{0}}
\newcommand{\bone}{\mathbf{1}}
\newcommand{\bbeta}{\bm{\beta}}

\newcommand{\var}{\text{var}}
\newcommand{\tdk}{\widetilde D_k}
\newcommand{\drw}{\hat w^{\text{DR}}}
\newcommand{\wsw}{\hat w^{\text{WS}}}
\newcommand{\csw}{\hat w^{\text{CS}}}
\newcommand{\drU}{\hat U^{\text{DR}}}
\newcommand{\wsU}{\hat U^{\text{WS}}}
\newcommand{\csU}{\hat U^{\text{CS}}}
\newcommand{\drY}{\hat Y^{\text{DR}}}
\newcommand{\csY}{\hat Y^{\text{CS}}}
\newcommand{\wsY}{\hat Y^{\text{WS}}}
\newcommand{\drS}{\hat \Sigma^{\text{DR}}}
\newcommand{\wsS}{\hat \Sigma^{\text{WS}}}
\newcommand{\csS}{\hat \Sigma^{\text{CS}}}

\newcommand{\drbb}{\hat b^{\text{DR}}}
\newcommand{\wsb}{\hat b^{\text{WS}}}
\newcommand{\csb}{\hat b^{\text{CS}}}

\newcommand{\drp}{\hat \psi^{\text{DR}}}
\newcommand{\csp}{\hat \psi^{\text{CS}}}

\newgeometry{left=2.8cm,right=2.8cm,bottom=1.5cm,top=1.0cm}
\title{Cross-validation Approaches for Multi-study Predictions}
\author[1,2]{Boyu Ren}
\author[3]{Prasad Patil}
\author[4]{Francesca Dominici}
\author[4,5]{Giovanni Parmigiani}
\author[4,5]{Lorenzo Trippa}
\affil[1]{Laboratory for Psychiatric Biostatistics, McLean Hospital}
\affil[2]{Department of Psychiatry, Harvard Medical School}
\affil[3]{Department of Biostatistics, Boston University School of Public Health}
\affil[4]{Department of Biostatistics, Harvard T.H. Chan School of Public Health}
\affil[5]{Department of Data Sciences, Dana-Farber Cancer Institute}
\date{}

\begin{document}
	
	\maketitle
	
	\begin{abstract}
		We consider prediction in multiple studies with potential differences in the relationships between predictors and outcomes. Our objective is to integrate data from multiple studies to develop prediction models for unseen studies. We propose and investigate two cross-validation approaches applicable to {\it multi-study stacking}, an ensemble method that linearly combines study-specific ensemble members to produce generalizable predictions. Among our cross-validation approaches are some that avoid reuse of the same data in both the training and stacking steps, as done in earlier multi-study stacking. We prove that under mild regularity conditions the proposed cross-validation approaches produce stacked prediction functions with oracle properties. We also identify analytically in which scenarios  the proposed  cross-validation approaches increase  prediction accuracy  compared to stacking with data reuse. We perform a simulation study to  illustrate these results. Finally, we apply our method to predicting mortality from long-term exposure to air pollutants, using collections of datasets.
	\end{abstract}
	
	\newpage

\section{Introduction}

It is increasingly common for researchers to have access to multiple studies and datasets to address the same or closely related prediction questions \citep{kannan2016public,manzoni2018genome}. 
Although datasets from $K>1$ studies may contain the same outcome variable $Y$ (e.g., patient survival) and predictors $X$ (e.g., pre-treatment prognostic profiles in clinical studies), the joint distributions $P_1, \ldots, P_K$ of $(X, Y)$ across these $K$ studies are typically different, owing to distinct study populations, study designs and study-specific technological artifacts \citep{patil2015test,sinha2017assessment}. In this article, we consider building \textit{prediction functions} (PFs) for future unseen studies using multiple datasets, while accounting for potential differences in the study-specific distributions $P_1, \ldots, P_K$. We consider $P_k$'s to be random objects and the heterogeneity across them to be described by a hyper-distribution of distributions.

We use the term generalist, as in \cite{bernau2014cross}, to indicate that we wish to make prediction in studies $K+1, K+2, \ldots$ that are not included in the training data collection. Strategies to develop a generalist PF depend on relations and similarities between studies. If studies can be considered exchangeable, i.e. $P_k$ can be assumed to be invariant to permutations of the study indices $1, \ldots, K$, then a model that consistently predicts accurately across all $K$ training studies is a good candidate for generalist predictions in unseen studies $k > K$. Generalist PFs have been studied from meta-analytic perspectives \citep{tseng2012comprehensive,pasolli2016machine} and using hierarchical models \citep{ventz2020integration,rashid2020modeling}, often assuming study exchangeability. If the exchangeability assumption is inadequate, joint models for multiple studies can incorporate information on relevant relations between studies \citep{moreno2012unifying}. For example, when $K$ datasets are representative of different time points, % $t_1 < t_2 < \ldots < t_K$, 
one can incorporate potential cycles or long-term trends.

The goal of generalist prediction is akin to that of domain generalization (DG) \citep{wang2021generalizing} in machine learning, if each study can be thought of as a domain, that is, a factor known to substantially affect the data generating distribution $P_k$. One of the goals of DG is to leverage the information of all $K$ observed domains to construct a PF that can generalize to an unseen target domain. Some algorithms in DG use representation learning to identify transformations of $X$ to serve as the predicting features for better generalizability. One strategy, domain invariant representation learning, finds such transformations of $X$ whose distributions are invariant across domains \citep{shao2019multi,deng2020representation}. Another, feature disentanglement, relaxes the assumption that invariant representations exist and instead learns transformations of $X$ that later are identified as domain-shared and domain-specific features. This strategy uses only the domain-shared features for predictions while minimizing the information loss of ignoring the domain-specific features \citep{ilse2020diva,wang2021variational}.

Recently, ensemble methods have been proposed for DG \citep{nozza2016deep,ding2017deep,zhou2020domain}, as well as multi-study learning \citep{patil2018training,Ramchandran:2019ga}. An important building block for these methods is stacking \citep{wolpert1992stacked,breiman1996stacked}. In a single study setting, stacking combines an ensemble of PFs, each trained in a subset of the available data, into a single PF. The weights assigned to each model are often selected by maximizing a utility function representative of the accuracy of the combined PF. Stacking is computationally efficient and finds the optimal PF within the convex hull of the ensemble \citep{juditsky2008learning,yao2018using}. When PFs are derived from flexible machine learning algorithms, their convex hull contains a large class of models, 
and hence stacking is more likely to recover the true data generating mechanism \citep{yao2018using}.

\cite{patil2018training} extended stacking to the multi-study setting via two-stage algorithms that, in the first stage, train learners on individual studies to generate a collection of single-study prediction functions (SPFs) with arbitrary machine learning methods. Then, at the second stage, they select ensemble weights using stacking on a merged dataset including all labels and all predictions. This rewards SPFs that do well outside their training studies. This approach does not require exchangeability, and the derivation of the ensemble weights can be tailored to settings where exchangeability is implausible. \cite{Loewinger856385} further extended this method to studies generated by resampling of the $K$ training datasets. In this article we formulate the second stage of the multi-study stacking approach as a formal optimization wherein stacking weights approximately maximize an expected utility function. The expectation is estimated using the entire collection of $K$ training studies. Our contribution is two-fold. 

First, we introduce cross-validation (CV) procedures specific for the multi-study settings, with the aim to mitigate the potential over-fitting associated with data reuse (DR) in two-stage stacking. This issue arises from the partial overlap of the data points used to train each SPF and to estimate stacking weights. Depending on how a multi-study dataset is partitioned into different folds for training SPFs and estimating weights, we propose two CV procedures: within-study CV ($\cvws$) and cross-study CV ($\cvcs$). $\cvws$ uses a subset of the data in each study to form a fold while $\cvcs$ treats each study as a fold. This approach shares similarities with the multi-source CV strategy discussed in \cite{geras}, but with a key difference: at each iteration, $\cvcs$ evaluate the out-of-study prediction accuracy of a weighted combination of SPFs, instead of a re-trained PF based on the aggregated training studies, to explicitly accommodates the utility function defined for multi-study stacking. We characterize the behavior of $\cvws$ and $\cvcs$ when applied for generalist predictions, with particular focus on the bias and variance of the estimated expected utilities \citep{bengio2003no}.

Second, we derive asymptotic results and provide empirical comparisons between the multi-study stacking PFs trained with DR and CV procedures. We evaluate these procedures using mean squared error (MSE) of prediction. Ours is the first theoretical investigation of the two-stage stacking with CV in multi-study settings. Our results supplement generalization error bounds for DG \citep{blanchard2021domain} and show that, when the number of  studies $K$ and the sample sizes $n_k$ become large, both DR and CV stacking achieve a generalization error similar to an asymptotic oracle benchmark. The asymptotic oracle is defined as the linear combination of the  SPFs' limits $\lim\limits_{n_k\to\infty} \hat Y_ k$ for $k=1,\ldots,K$ that  minimizes the MSE in future  studies, that is for $k > K$. Our results bound the MSE difference between the oracle ensemble and two stacking procedures, with and without DR. Related bounds have been studied in the single-study setting \citep{vdl1} and in the functional aggregation literature \citep{juditsky2000functional,juditsky2008learning,transfer-li}.

We apply our DR and CV stacking procedures to predicting pollution-related mortality. We use information on Medicare beneficiaries and measurements of air pollutants at the ZIP code level. We are interested in predicting the number of deaths per 10,000 person-years. In separate analyses, we partitioned the database into studies defined by state and by county. We compare the relative performance of DR and CV stacking.

\section{Generalist predictions}

\subsection{Notation}
We use data from $K$ studies $k = 1, \ldots, K$, with sample sizes $n_k$. For individual $i$ in study $k$ we have a vector of $p$-dimensional characteristics $x_{i,k} \in \mathcal X \subseteq \mathbb R^p$ and the individual outcome $y_{i,k} \in \mathbb R$. Let $\mathcal S_k = \{(x_{i,k},y_{i,k}),i=1,\ldots,n_k\}$ denote all data from study $k$ and $\mathcal S = \{\mathcal S_k, k = 1, \ldots, K\}$ the collection of all $K$ datasets. Extending earlier multi-study architectures, we define a list $\mathcal D$ of training sets, which includes $T$ disjoint members $D_1,\ldots ,D_T$. Each $D_t$ is a set of $(i, k)$ indices, where $i \in \{1, 2, \ldots, n_k\}$ is the sample index within a study $k$. The set $D_t$ can include indices with different $k$ values (see for example $D_2$ and $D_3$ in Figure \ref{fig:D-example}). We call a collection $\mathcal D$ \textit{partitioned by study} if $T = K$ and $D_t=\widetilde D_t = \{(1,t),\ldots,(n_t,t)\}$ for $t=1,\ldots,K$.

\begin{figure}[htbp]
	\centering
	\includegraphics[scale=0.3]{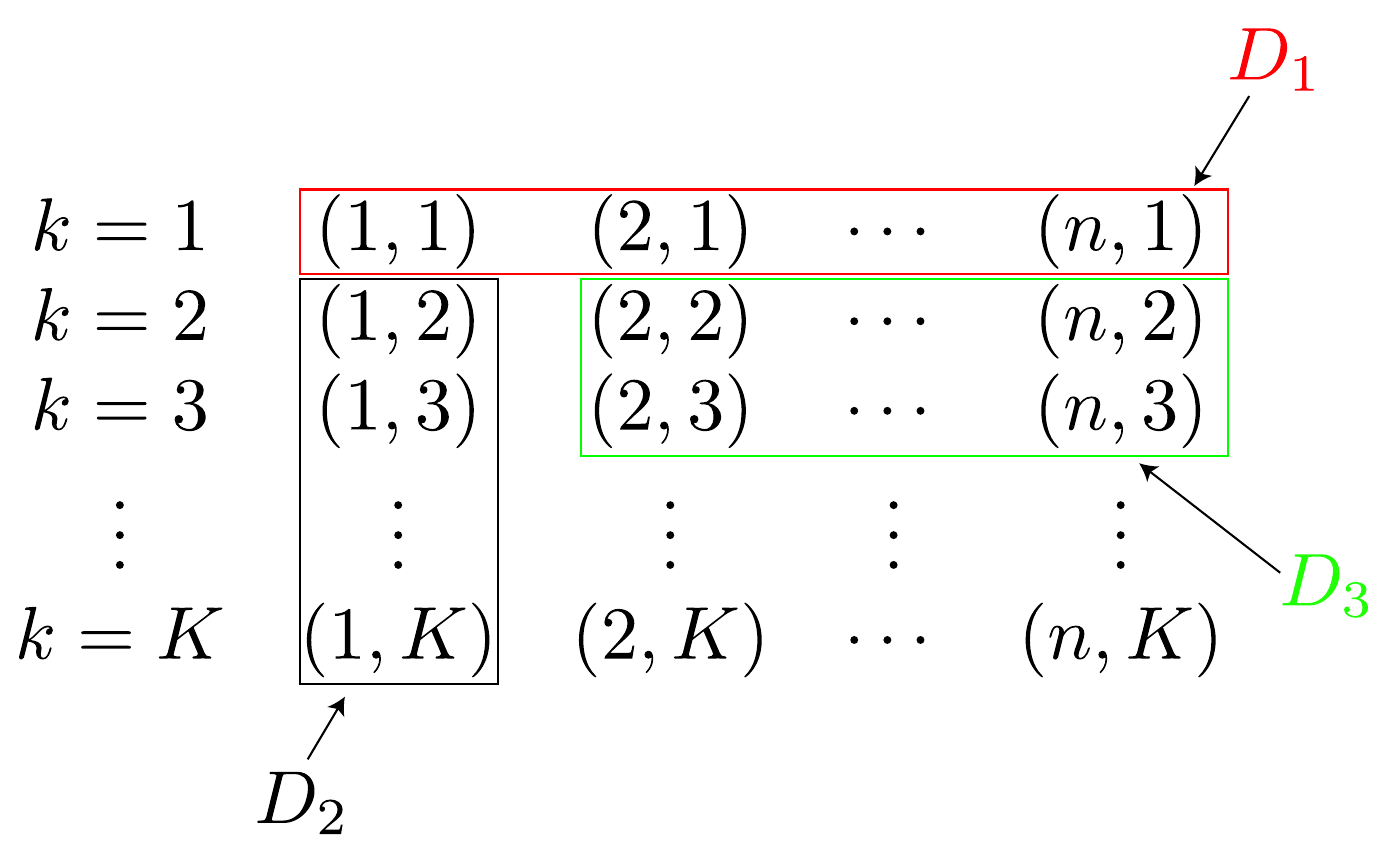}
	\caption{\small \it Illustration of the relation between studies and training sets for $\mathcal D=\{D_1,D_2,D_3\}$ when all studies are of equal size ($n_k=n$, $k = 1,\ldots,K$).}
	\label{fig:D-example}
\end{figure}

We consider $L$ different learners. A learner is a method for generating a PF, such as linear regression, random forests or a neural network. Note that we do not require a learner to have the ability to assess the evidential uncertainty associated with its output PF. Training the learner $\ell$ on $D_t$ generates an SPF denoted as $\hat Y_t^\ell: \mathcal X \to \mathbb R$. The set of  SPFs is $\hat{\mathcal Y} = \{\hat Y_t^\ell(\cdot);\ell=1,\ldots,L,\;t = 1,\ldots,T\}$. Let $W$ be a subset of $\mathbb R^{TL}$, where $\mathbb R^{TL}$ is the $TL$-dimensional Euclidean space. We combine the SPFs in $\hat{\mathcal Y}$ into $\hat Y_w:\mathcal X \to \mathbb R$,
$$
\hat Y_w (\cdot) = \sum_{\ell=1}^L\sum_{t=1}^T w_{\ell, t} \hat Y_t^\ell(\cdot),
$$
where $w = (w_{\ell,t};\ell\leq L, t\leq T)\in W$ are weights trained using multi-study stacking. In the remainder of this manuscript, if not otherwise specified, we suppress the range of the summations and assume that the range for $i$ is $1$ to $n_k$, 1 to $K$ for $k$, $1$ to $T$ for $t$, and $1$ to $L$ for $\ell$.

\subsection{Utility function for generalist prediction}
We want to use $\hat Y_w$ for prediction in a target population with unknown joint distribution $\pi$ on
$(X, Y)$. The objective is to maximize the expected
utility $U$ of $\hat Y_w$ in the target population:
$$
U(w;\pi) = \int_{(x,y)} u(\hat Y_w(x), y) \; d\pi(x,y),
$$
where $u(y, y')$ is a utility function, such as the negative squared difference $u(y, y') = -(y - y')^2$.

Generalist predictions are designed to target \textit{all} future studies $K + 1, K + 2, \ldots$. The corresponding utility function is defined through a limit:
%ing average utility of a generalist PF $\hat Y_w$ over repeated use with no additional training is
$$
U_g(w) = \lim_{I\to\infty}\frac1I \sum_{i=1}^I U(w;P_{K+i}).
$$
We will only consider scenarios where the above limit is well defined. In particular, if $P_1,P_2,\ldots$ are exchangeable, there exists a fixed distribution $Q$ of random distributions such that $P_k\overset{iid}{\sim} Q$ for $k = 1, 2, \ldots.$ In this case, $U_g(w)$ can be rewritten, for a generic $P \sim Q$,  as
$$
U_g(w) = \int \left[\int_{(x,y)} u(\hat Y_w(x), y)\diff P(x,y)\right]\diff Q.
$$
Changing the order of integration, and defining $P_0(\cdot) = \int P(\cdot)dQ$, 
$$
U_g(w) = U(w;P_0)= \int_{(x,y)} u(\hat Y_w(x), y) \diff P_0(x,y).
$$
Note that when exchangeability across $P_k$ holds, $P_0$ can be estimated with the empirical distribution $\hat P_0(x,y) = K^{-1}\sum_{k,i} \mathbb I(x=x_{i,k},y=y_{i,k})/n_k$. Hence, a direct estimator of $U_g(w)$ is
$$
\drU(w) = \sum_{k} \frac{1}{K}\cdot\frac{1}{n_k}\sum_{i} u(\hat Y_w(x_{i,k}), y_{i,k}),
$$
where DR indicates that the estimate of $U$ involves DR. See Section \ref{sec:stacking} for a discussion of DR.

If the sequence $P_1,P_2\ldots$ is not exchangeable, for example when the sequence satisfies a Markov relation $[P_k|P_1,\ldots,P_{k-1}]\overset{d}{=} [P_k|P_{k-1}]$ as in \cite{shumway2017time}, we can replace $1/K$ by weights $\nu_k$ varying with $k$ (e.g., $\nu_k$ is larger when $k$ gets closer to $K$). Throughout this article we will not specify $Q$, but we will assume exchangeability of $P_1,P_2,\ldots$.

\section{Multi-study Stacking Approaches}
\label{sec:stacking}
In this section, we introduce three approaches to train generalist PFs in the multi-study setting. We begin by defining the {\it oracle} weights as
$$
w_g = \argmax_{w\in W} U_g(w).
$$
In what follows, if not otherwise specified, $u(y,y') = -(y - y')^2$. For the oracle case the generalist prediction error of the PF $\hat Y$ is
\begin{equation}
	\text{MSE}_0{(\hat Y)} = \int_{(x,y)} \left(y - \hat Y(x)\right)^2 dP_0.
	\label{eq:cMSE-def}
\end{equation}

% Note that the same optimization goals for $w$ can often be achieved by applying constraints or penalties. In several optimization problems an equivalence holds between constraining a $K\times L$-dimensional parameter, (e.g.\ $w\in W = \{w:\|w\|_2\leq c\}$, where $\|\cdot\|_q, q>0$ is the $L_q$ norm in the Euclidean space) and imposing an $L_2$ penalty on the unconstrained optimization. The use of penalties in the estimation of stacking weights has been discussed in \cite{breiman1996stacked} and \cite{leblanc1996combining}.

We consider three multi-study stacking approaches: the original multi-study stacking \citep{patil2018training}, which involves data reuse (\textit{DR stacking}, for short) and two novel CV stacking approaches: \textit{within-study cross-validated stacking} ($\cvws$) and \textit{cross-study cross-validated stacking} ($\cvcs$). These three approaches select $w$ by maximizing different estimates of the expected utility $U_g(w)$. We discuss the estimators of $U_g(w)$ underlying these approaches. We study the accuracy of an estimator $\hat U(w)$ by examining its deviation from the truth $\hat U(w) - U_g(w)$ and we focus mainly on the mean $\mathbb E_{\mathcal S}(\hat U(w) - U_g(w))$ and variance $\var_{\mathcal S}(\hat U(w) - U_g(w))$ over different realizations of $\mathcal S$, the data from the entire collection of $K$ studies. We refer to $\mathbb E_{\mathcal S}(\hat U(w) - U_g(w))$ as the \textit{bias} of $\hat U(w)$.

In two examples we illustrate that DR stacking and $\cvws$ tend to over-estimate the utility function $U_g$ ($\mathbb E_{\mathcal S}(\hat U(w) - U_g(w))>0$) and have smaller variance compared to $\cvcs$, which in turn tends to under-estimate $U_g$.

\subsection{DR stacking}
\label{sec:dr-stacking}
A direct approach to selecting $w$ for generalist predictions is to maximize the estimate $\drU(w)$ of $U_g(w)$. Let $\drw = \argmax_{w\in W} \drU(w)$ and $\drY = \sum_{\ell, t} \drw_{\ell,t}\hat Y_t^\ell$.
Training the SPFs $\hat Y^{\ell}_t$ uses data in $D_t$, which in this case is later re-used to compute $\drU(w)$. This could lead to a non-zero bias $\mathbb E_{\mathcal S}(\drU(w) - U_g(w))$ that may also depend on $w$. To see this, note that the difference between $U_g(w)$ and $\drU(w)$ can be expressed as 
$$
\drU(w) - U_g(w) = 
- w^\intercal (\drS - \Sigma) w + 2 w^\intercal (\drbb - b) + C,
$$
where $\mathbb E_{\mathcal S}(C) =0$. Here $\drS$ and $\Sigma$ are $TL\times TL$ matrices and their components corresponding to $w_{t,\ell}\times w_{t',\ell'}$ are
$$
\drS_{t,t',\ell,\ell'} = \sum_{k}\frac{1}{Kn_k}\sum_i\hat Y_{t}^\ell(x_{i,k})\hat Y_{t'}^{\ell'}(x_{i,k}), \text{and}~\Sigma_{t,t',\ell,\ell'} = \langle \hat Y_t^\ell, \hat Y_{t'}^{\ell'}\rangle,
$$
where $\langle\hat Y_t^\ell, \hat Y_{t'}^{\ell'}\rangle = \mathbb E_{x\sim P_0}(\hat Y_t^\ell(x)\hat Y_{t'}^{\ell'}(x)|\mathcal S)$. Also, $\drbb$ and $b$ are $TL$-dimensional vectors:
$$
\drbb_{k,\ell} = \sum_{i,k}\frac{1}{Kn_k}\hat Y^\ell_t(x_{i,k})y_{i,k}, \text{and}~b_{t,\ell} = \langle \hat Y_t^\ell, Y_0 \rangle,
$$
where $\langle \hat Y_t^\ell, Y_0\rangle = \mathbb E_{(x,y)\sim P_0}(y\hat Y_t^\ell(x)|\mathcal S)$. 
%When $K\geq 3$, one can construct unbiased estimators of $\Sigma_{k,k',\ell,\ell'}$ and $b_{k,\ell}$. For example, $(K-2)^{-1}\sum\limits_{k''\notin \{k,k'\}}n_k^{-1}\sum_i \hat Y_k^\ell(x_{i,k''})\hat Y_{k'}^{\ell'}(x_{i,k''})$ is an unbiased estimator of $\Sigma_{k,k',\ell,\ell'}$ for any $k\neq k'$.

Consider $\mathcal D$ partitioned by study. If $k\notin\{t,t'\}$, then
$$
\mathbb E_\mathcal S\left(\hat Y_{t}^\ell(x_{i,k})\hat Y_{t'}^{\ell'}(x_{i,k})|\mathcal S_t,\mathcal S_{t'}\right) = \langle \hat Y_{t}^\ell, \hat Y_{t'}^{\ell'}\rangle, \text{and}~
\mathbb E_\mathcal S\left(\hat Y_{t}^\ell(x_{i,k})y_{i,k}|\mathcal S_t\right) = \langle \hat Y_{t}^\ell, Y_0\rangle.
$$
If $k\in\{t,t'\}$, the first equality does not hold since conditioning on $\mathcal S_t$ and $\mathcal S_{t'}$, $\hat Y_{t}^\ell(x_{i,k})\hat Y_{t'}^{\ell'}(x_{i,k})$ is a constant in the expectation, and this constant depends on $x_{i,k}$. Hence, it cannot be equal to $\langle \hat Y_t^\ell, \hat Y_{t'}^{\ell'}\rangle$, which is independent of $x_{i,k}$. Similarly, if $k = t$, the second equality does not hold. Specifically,
$$
\mathbb E_\mathcal S\left(\drS_{t,t',\ell,\ell'} - \Sigma_{t,t',\ell,\ell'}\big| \mathcal S_t,\mathcal S_{t'}\right) = \sum_{k\in \{t,t'\}} \sum_{i}\frac{\left( \hat Y_{t}^\ell(x_{i,k})\hat Y_{t'}^{\ell'}(x_{i,k}) - \langle \hat Y_{t}^\ell, \hat Y_{t'}^{\ell'}\rangle\right)}{Kn_k},$$ 
and
$$
\mathbb E_\mathcal S\left(\hat b^{\text{DR}}_{t,\ell} - b_{t,\ell}\big|\mathcal S_t \right) = \sum_{i}\frac{1}{Kn_t}\left(\hat Y^\ell_t(x_{i,t})y_{i,t} - \langle \hat Y_t^\ell, Y_0 \rangle\right).
$$
Integrating out $\mathcal S_t$ and $\mathcal S_{t'}$ in the first equation, we have
$$
\mathbb E_\mathcal S\left(\drS_{t,t',\ell,\ell'} - \Sigma_{t,t',\ell,\ell'}\right) = \sum\limits_{k\in \{t,t'\}}\sum_{i}\frac{(\mathbb E( \hat Y_{t}^\ell(x_{i,k})\hat Y_{t'}^{\ell'}(x_{i,k})) - \mathbb E\langle \hat Y_{t}^\ell, \hat Y_{t'}^{\ell'}\rangle)}{Kn_k} .
$$
At least one of $\hat Y_t^\ell$ and $\hat Y_{t'}^{\ell'}$ has been trained with data including $x_{i,k}$. Hence there is no guarantee that the expectation above is zero as it involves a difference between ``in-sample'' estimate $\hat Y_{t}^\ell(x_{i,k})\hat Y_{t'}^{\ell'}(x_{i,k})$ of an ``out-of-sample'' target $\langle \hat Y_{t}^\ell, \hat Y_{t'}^{\ell'}\rangle$. The same result applies to $\drbb$, which implies that $\drU$ is potentially biased.

\begin{example}[No predictors, DR]
	\label{ex:two-y}
	This example illustrates key points using two studies and no predictors. 
	Here the oracle solution favors the SPF with higher generalist prediction accuracy, while DR stacking does not. Also, DR over-estimates $U_g$, and its bias grows with inter-study heterogeneity. 
	We will revisit this example to discuss various stacking procedures throughout the manuscript. See Appendix for detailed derivations.
	
	Consider a $\mathcal D$ partitioned by study with $K=2$. Assume we observe outcomes $y_{i,k}\sim N(\mu_k,1)$ for $k=1,2$ and $n_1=n_2 = n$. Across studies, $\mu_k\sim N(0,\sigma^2)$. Consider the two SPFs  $\hat Y_1(\cdot) = \bar y_1$ and $\hat Y_2(\cdot) = \bar y_2$, where $\bar y_k = n_k^{-1}\sum_i y_{i,k}$. Set $W = \Delta_1$, the 1-simplex.
	
	\noindent \underline{\it Oracle weights $w_g$}. The oracle weight for the SPF $\hat Y_1$, $w_{g;1}$, is 
	$$
	w_{g;1} = \left\{\begin{array}{cc}
		\frac{|\bar y_2|}{|\bar y_1| + |\bar y_2|} & \bar y_1\cdot \bar y_2<0.\\
		1 & \bar y_1\cdot \bar y_2\geq 0, |\bar y_1|\leq |\bar y_2|\\
		0 & \bar y_1\cdot \bar y_2\geq 0, |\bar y_1| > |\bar y_2|
		\end{array}\right..
	$$
	Thus, the oracle favors $\hat Y_1$, i.e. $w_{g;1}>w_{g;2}$, whenever $\text{MSE}_0(\hat Y_1) < \text{MSE}_0(\hat Y_2)$.
	
	\noindent \underline{\it DR stacking weights}. The DR stacking weights turn out to be $\drw = (1/2,1/2)$. In other words, DR stacking does not place more weight on the PF with smaller $\text{MSE}_0$. Also, the mean and variance of $\drU(w) - U_g(w)$ are
	\begin{gather*}
	\mathbb E_\mathcal S\left(\drU(w) - U_g(w)\right)
	= \sigma^2+1/n,\\
	\var_\mathcal S\left(\drU(w) - U_g(w)\right) = (\sigma^2 + 1/n)^2\left(1 + 2(w_1^2+w_2^2)\right).
	\end{gather*}
	Thus $\drU(w)$ over-estimates $U_g(w)$. Bias and variance remain positive when $n\to\infty$ and are increasing functions of $\sigma^2$, which quantifies the inter-study heterogeneity.
\end{example}

\subsection{CV stacking}
\label{sec:zero-stacking}
In single-study stacking, CV is implemented by using part of the data for the training of the PFs $\hat{\mathcal Y}$ and the rest for the selection of $w$ (see for example \cite{breiman1996stacked}). We introduce two approaches to generalize CV to the multi-study setting.
{\bf Within-study CV ($\cvws$)}, in Section~\ref{sec:cvws}, partitions each study in folds. For each learner, one fold is set aside for selecting the multi-study stacking weight, while the others are used for training SPFs.
{\bf Cross-set CV ($\cvcs$)}, in Section~\ref{sec:cvcs}, treats studies as folds. At iteration $k = 1,\ldots, K$, we build the library of SPFs using all $\mathcal D_t$ that do not contain samples from study $k$. We then combine this restricted library of SPFs to predict for study $k$. The optimal $w$ maximizes a utility estimate based on the predictions generated across all $K$ iterations.

\subsubsection{Within-study CV}
\label{sec:cvws}

This approach includes $M$ iterations and the four steps below. For simplicity, we assume that $|\tdk|$ is divisible by $M$ for $k=1,\ldots, K$, where $|D|$ is the cardinality of $D$.

\begin{enumerate}
	\item Randomly partition each index set $\tdk$ into $M$ equal-size subsets $\widetilde D_{k,1},\ldots, \allowbreak\widetilde D_{k,M}$.
	\item For every $m=1,\ldots,M$, we train SPFs $\hat Y_{t,m}^\ell$ using data in $\mathcal D_t$ that is not included in the $m$-th fold, that is, $\{(x_{i,k},y_{i,k});(i,k)\in D_{t}\cap (\bigcup_{k,m'\neq m} \widetilde D_{k,m'})\}$, for $\ell = 1,\ldots, L$ and $t = 1,\ldots,T$.
	\item For each sample $(i,k)$, denote by $m(i,k)$ the single index $m$ such that $(i,k)\in \widetilde D_{k,m}$. The estimated utility function for the generalist predictions is
	$$
	\wsU(w) = \sum_{k}\frac{1}{Kn_k}\sum_{i}u\left(\sum_{\ell,t}w_{\ell,t}\hat Y_{t,m(i,k)}^\ell(x_{i,k}), y_{i,k}\right),
	$$
	and the optimal weights solve
	$$
	\wsw = \argmax_{w\in W} \wsU(w).
	$$
	\item The $\cvws$ stacked PF is
	$
	\wsY = \sum_{\ell, t} \wsw_{\ell,t}\hat Y_t^\ell.
	$
\end{enumerate}
We could also repeat steps 1-3 for multiple random partitions of the data collection $\mathcal S$ into $M$ folds. Each repetition results in an estimate $\wsU(w)$. Although each repetition generates a similar but not identical set of SPFs, it may still be useful to consider consensus weights $\wsU$ obtained, for example, by optimizing the average over all the estimated utilities.

Although $\cvws$ seemingly avoids DR through holding out part of $\mathcal S$ for the selection of $w$, this strategy turns out to be insufficient in many cases and leads to a utility estimator $\wsU$ that is almost identical to $\drU$. We illustrate this phenomenon in the setting of Example \ref{ex:two-y} and in a regression setting (see Appendix for derivations). 

\setcounter{example}{0}
\begin{example}[No predictors, $\cvws$]
	When applying $\cvws$ in this example, we have
	$$
	|\wsU(w)-\drU(w)| = o_p(|\drU(w)-\lim_{n\to\infty}\drU(w)|).
	$$
	Specifically, we can show that $|\wsU(w) - \drU(w)| = O_p(1/n)$. In contrast, $\sqrt{n}((\bar y_1, \bar y_2)^\intercal - \mu) \overset{d}{\to}N(0,I_2)$ as $n\to\infty$, where $I_2$ is the identify matrix, and therefore $\drU(w) - \lim\limits_{n\to\infty}\drU(w) = O_p(1/\sqrt{n})$. This result suggests with moderate sample size $n$, the difference between the two utility estimators is of smaller magnitude than the variability of the DR stacking estimator itself, and the two estimators can be regarded as identical in practice.
\end{example}

\begin{example}[Regression, $\cvws$ vs. DR]
	\label{ex:K-regression}
	Consider a scenario with $K$ studies, and
	\begin{equation}
		\begin{gathered}
			y_{i,k} = \beta_k^\intercal x_{i,k} + \epsilon_{i,k},~~\beta_k \sim N(\beta_0, \sigma_\beta^2 I_p),~~\epsilon_{i,k}\sim N(0,\sigma^2),
		\end{gathered}
		\label{eq:diff-sim}
	\end{equation}
	where $\beta_k\in\mathbb R^p$ are study-specific regression coefficients and $\epsilon_{i,k}$ is a noise term. We also assume that the $p-$dimensional vector $x_{i,k}$ follows a $N(0,I_p)$ distribution. We prove that again, $|\wsU(w) - \drU(w)| = o_p(|\drU(w) - \lim\limits_{n\to\infty}\drU(w)|)$ as $n\to\infty$ in Proposition \ref{prop:cv-eq}.
\end{example}

\begin{proposition}
	\label{prop:cv-eq}
	Assume $\mathcal D$ is partitioned by study and $n_k=n$ for $k\in\{1,2,\ldots,\allowbreak K\}$. Fix $\beta_1,\ldots,\beta_K$ in model (\ref{eq:diff-sim}). Let $L=1$ and SPFs $\hat Y^\ell_k$ be OLS regression functions. For any $w\in W$, where $W$ is a bounded set in $\mathbb R^{K}$, the following results hold:
	\begin{gather*}
		\sup_{w\in W}\left|\drU(w) - \wsU(w)\right| = O_p(1/n),\\
		\sqrt{n}\left(\left|\drU(w) - \lim\limits_{n\to\infty}\drU(w)\right|\right) \overset{d}{\to} Z.
	\end{gather*}
	Here $Z$ a non-degenerate normally distributed random variable.
\end{proposition}

In Figure \ref{fig:sec3-main}(a), we plot $\left|\drU(w) - \wsU(w)\right|$ and $\left|\drU(w) - \lim\limits_{n\to\infty} \drU(w)\right|$ at $w = \bone_K/K$ as a function of $n$. We find that $\log\left|\drU(w) - \wsU(w)\right|$ can be approximated by $\left(-\log n + c_0\right)$ while $\log \left|\drU(w) - \lim\limits_{n\to\infty} \drU(w)\right|$ can be approximated by $\left(-0.5\log n + c_1\right)$, where $c_0$ and $c_1$ are constants. The results are concordant with the rates of convergence in Proposition \ref{prop:cv-eq}.

\subsubsection{Cross-set CV}
\label{sec:cvcs}
In this section, we introduce $\cvcs$ by focusing on the special case of leave-one-study-out $\cvcs$ and when $\mathcal D$ is partitioned by study. The estimator of $U_g(w)$ is constructed as follows.
\begin{enumerate}
	\item Generate the library of SPFs $\hat{\mathcal Y}$ using every set in $\mathcal D$ (i.e. every study in this case). This library remains identical across the $K$ iterations.
	\item At iterations $k = 1,\ldots, K$, evaluate the utility of $w$ using $\tdk$ as validation set, excluding SPFs trained on $\tdk$:
	\begin{equation}
		\csU_k(w) = \frac{1}{n_k}\sum_{i} u\left(\sum_{\ell,k'}\mathbb I(k'\neq k) w_{\ell,k'} \hat Y_{k'}^\ell(x_{i,k}), y_{i,k}\right).
		\label{eq:loo-ind}
	\end{equation}
	\item Obtain the utility estimator $\csU (w)$ by combining all $\csU_k$ across the $K$ iterations:
	\begin{equation}
		\csU (w)  = \frac{1}{K}\sum_{k} \csU_k\left(\frac{w}{1-1/K}\right).
		\label{eq:cvcs-U-study}
	\end{equation}
	
	\item  Select $w$ as $\csw = \argmax_{w\in W} \csU (w)$. The $\cvcs$ stacked PF is
	$
	\csY = \sum_{\ell, k} \csw_{\ell, k}\hat Y_{k}^\ell.
	$
\end{enumerate}
%In each iteration of $\cvcs$ ($k=1,\ldots,K$), the SPFs that are combined (cf. equation (\ref{eq:loo-ind})) originate only from studies that do not overlap with $\tdk$, which is used as the validation set in each iteration to compute $\csU_k(w)$.

\noindent \textbf{Interpretation of $\cvcs$.} Expression (\ref{eq:cvcs-U-study}) and  $\csY$ are simple to interpret when each study is regarded as a ``fold''. Our goal is to obtain an accurate estimate of the generalist utility for stacking weights. Similar to single-study CV, $\cvcs$ will sequentially hold out each study (fold) and use it to test the stacked PF's generalist prediction performance. When study $k$ is held out, SPFs $\hat Y_k^\ell$, $\ell = 1,\ldots,L$ from study $k$ are removed from the stacking stage to avoid data reuse. The stacked PF using this reduced list of SPFs can approximate $\hat Y_w$ by rescaling $w$. If we assume that the total contribution of all SPFs of a study $k$ to the final prediction is roughly identical across $k = 1,\ldots,K$, which is not too unreasonable given that $P_1,P_2,\ldots$ are exchangeable, then $\hat Y_w \approx \sum_{k'\neq k,\ell} w_{\ell, k'}/(1-1/K)\hat Y_{k'}^\ell$. The estimated generalist utility for the original library of SPFs with study $k$ as the hold-out is thus $\csU_k(w/(1-1/K))$. The overall estimated generalist utility, after iterating over all studies, is the average of $\csU_k(w/(1-1/K))$, as defined in (\ref{eq:cvcs-U-study}). See the following paragraph for a more concrete discussion.\\

%if we consider the following competition. A player will predict in population $k=K+1$  using  $\hat Y_w = \sum_{\ell,k}  w_{\ell,k} \hat Y_{k}^{\ell}$. However, the  SPFs  from one  of the studies  $k\le K$ will be eliminated by an adversary and the probability that the $k$-th set  of  SPFs  $\{\hat Y^\ell_k; \ell=1,\ldots, L\}$ becomes inaccessible is $1/K$. Our player knows  that  a  set of SPFs  will be eliminated and selects $(w_{\ell,k}; \ell=1,\ldots,L, k = 1,\ldots,K) \in W = \Delta_{KL-1}$, where $\Delta_{KL-1}$ is the ($KL-1$)-simplex, based on all $K$  studies and  SPFs $\hat{\mathcal{Y}}$. Since some SPFs $\hat Y_{k}^\ell$  will  be  lost,  the player will later update the  weights for these SPFs  to zero and the remaining ones  are normalized to remain in  the simplex. Expression (\ref{eq:cvcs-U-study}) provides for any $(w_{\ell,k};\ell=1,\ldots,L, k=1,\ldots,K) \in \Delta_{KL-1}$ an unbiased estimate of the  expected  utility.

\noindent \textbf{Rationale of $\cvcs$.} We motivate $\cvcs$ as an approximated unbiased estimator of $U_g$. Define $w_{-k}$ to be equal to $w$ except for components $w_{\ell,k}$, $\ell = 1,\ldots,L$, which are set to zero.
With exchangeable $P_k$ distributions, $\csU_k(w)$ is an unbiased estimator of $U_g(w_{-k})$, that is:
$$
\mathbb E\left[\csU_k(w)\mid \mathcal S_{-k}\right] = U_g(w_{-k}),
$$
where $\mathcal S_{-k} = \{\mathcal S_{k'};k'\neq k\}$.
It follows that 
$$
\sum_{k} \mathbb E\left[ \csU_k(w)\mid \mathcal S_{-k} \right]/K = \sum_{k=1}^K  U_g(w_{-k})/K.
$$
Consider the Taylor expansion of $\sum_k U_g(w_{-k})/K$ around $\bone_K/K$:
$$
\sum_k U_g(w_{-k})/K \approx U_g(\bone_K/K) + \frac{\partial U_g}{\partial w^\intercal}(\bone_K/K) \left(\sum_k w_{-k}/K - \bone_K/K\right).
$$
By construction, $\sum_k w_{-k}/K = (1-K^{-1})w$.
% ((1-K^{-1})w_{\ell,k};\ell=1,\ldots,L,k = 1,\ldots,K)$.
It follows that
\begin{align*}
\sum_k \mathbb E\left[\csU_k(w)\mid \mathcal S_{-k}\right]/K &\approx U_g(\bone_K/K) + \frac{\partial U_g}{\partial w^\intercal}(\bone_K/K)\left( (1-K^{-1})w-\bone_K/K \right) \\
&\approx U_g((1-K^{-1})w).
\end{align*}
The last approximation holds because the first approximation is the Taylor expansion of $U_g((1-K^{-1})w)$ at $w = \bone_K/K$ up to the first order. Note that $\csU(w) = \sum_k\csU_k\left(w/(1-K^{-1})\right)/K$, therefore
$$
\mathbb E(\csU(w)|\mathcal S) = K^{-1}\sum_k \mathbb E\left[\csU_k\left(\frac{w}{1-K^{-1}}\right)\bigg |\mathcal S_{-k}\right]\approx U_g(w),
$$
for $w$ close to $\bone_K/K$.\\

We next illustrate that the finite-sample behavior of $\cvcs$ stacking is distinct from that of DR stacking and $\cvws$ in Examples \ref{ex:two-y} and \ref{ex:K-regression}. We also introduce Example~\ref{ex:three-study} to illustrate the important differences between $\cvcs$ and $\cvws$ in generalist predictions.

\setcounter{example}{0}
\begin{example}[No predictors, $\cvcs$]
	We consider $\csU(w)$ for generalist predictions. See Appendix for detailed derivations of results below. In this example $\mathcal D$ is partitioned by study, and
	$$
	\csU(w) = -\left(2w_1^2\bar y_1^2 + 2w_2^2\bar y_2^2 - 2\bar y_1\bar y_2 + \frac{\overline{y_1^2}+\overline{y_2^2}}{2}\right),
	$$
	It follows that 
	\begin{gather*}
	\mathbb E\left(\csU(w) - U_g(w)\right) = -(\sigma^2+1/n)(w_1^2+w_2^2),\\
	\var\left(\csU(w) - U_g(w)\right) = 2(\sigma^2+1/n)^2\left(w_1^2 + w_2^2\right)^2.
	\end{gather*}
	Therefore, $\csU(w)$ underestimates $U_g(w)$. Also, for every $w\in\Delta_1$, $\csU$ has a smaller bias and variance compared to $\drU$:
	\begin{gather*}
	|\mathbb E\left(\csU(w) - U_g(w)\right)| < |\mathbb E\left(\drU(w) - U_g(w)\right)|,\\ \var\left(\csU(w)-U_g(w)\right) < \var\left(\drU(w)-U_g(w)\right).
	\end{gather*}
	
	The maximizer of $\csU(w)$ is
	$
	\csw = \left(\bar y_2^2/(\bar y_1^2+\bar y_2^2), \bar y_1^2/(\bar y_1^2 + \bar y_2^2)\right)$. Like the oracle weights $w_g$, $\csw$ favors the SPF with lower $\text{MSE}_0$ (see equation \refeq{eq:cMSE-def}). More importantly, $\cvcs$ has smaller generalist prediction error $\text{MSE}_0$ than DR stacking:
	$$
	\text{MSE}_0\left(\csY\right) - \text{MSE}_0\left(\drY\right) = \frac{-(\bar y_1^2 - \bar y_2^2)^2(\bar y_1 + \bar y_2)^2}{4(\bar y_1^2 + \bar y_2^2)^2}\leq 0.
	$$
	The equality holds if and only if $\bar y_1 = \bar y_2$.
\end{example}

\begin{example}[Regression, $\cvcs$]
	We revisit the regression setting to compare $\cvcs$, DR stacking and $\cvws$. As $n\to\infty$, $\cvcs$ behaves differently compared to $\cvws$. Under mild assumptions (cf.\ Prop.~\ref{prop:cv-eq}), $\wsU(w) - \drU(w)\overset{p}{\to}0$ as $n\to\infty$ when $K$ is fixed. This is not the case for $\cvcs$. As $n\to\infty$,we have
	$$
	\csU(w)\overset{p}{\to} \sum\limits_{k,k'} \frac{K(K-2+\delta_{k,k'})}{(K-1)^2}w_k\beta_k^\intercal \beta_{k'}w_{k'}-2\sum_k w_k\beta_k^\intercal\bar\beta_{-k} + \frac{\sum_k \beta_k^\intercal\beta_k}{K} + \sigma^2,
	$$
	where $\delta_{k,k'} = \mathbb I(k = k')$ and $\bar\beta_{-k} = \sum_{k'\neq k} \beta_{k'}/(K-1)$. This limit is
	different from the limit of $\drU$,
	$$
	\drU(w)\overset{p}{\to} \sum\limits_{k,k'}w_k\beta_k^\intercal\beta_{k'}w_{k'}-\frac{2\sum_kw_k\beta_k^\intercal(\sum_{k'}\beta_{k'})}{K} + \frac{\sum_k \beta_k^\intercal\beta_k}{K} + \sigma^2.
	$$
	
	Both DR stacking and $\cvws$ overestimate $U_g$. We compute by Monte Carlo simulations the bias and variance for DR stacking, $\cvws$ and $\cvcs$ (Fig. \ref{fig:sec3-main}(b-c)). We assume that $\sigma = \sigma_{\beta} = 1$, $K = 20$, $p=10$ and $n_k=100$. The only learner ($L = 1$) is OLS. The means and variances are derived based on 50 replicates. The data generating model is defined by expression (\ref{eq:diff-sim}), and we generate a new set of $\beta_1,\ldots,\beta_K$ in each simulation replicate. The bias and variance for DR and $\cvws$ are nearly identical (cf.\ Prop.\ \ref{prop:cv-eq}). In \ref{fig:sec3-main}(c), the curves for DR and $\cvws$ even overlap with each other. On the other hand, $\cvcs$ has smaller bias but larger variance compared to DR stacking and $\cvws$, indicating a bias-variance trade-off of the $\cvcs$ strategy.
\end{example}

\begin{example}
	\label{ex:three-study}
	We consider the data generating model in (\ref{eq:diff-sim}), with a different distribution of $\beta_k$'s such that
	$
	\beta_k \overset{iid}{\sim} 0.5\delta_{\beta_0} + 0.5 N(\beta_0, \sigma^2_\beta I_p),
	$
	where $\delta_{\beta_0}$ is the point-mass distribution at $\beta_0$. The oracle PF is $Y_g(x) = \beta_0^\intercal x$. Assume $K = 3$, $n_k = 10,000$, $k=1,2,3$, and $\sigma^2 = 10$. In a single simulation of the multi-study collection $\mathcal S$, we set $\beta_1=\beta_2=\beta_0$ and $\beta_3 \neq \beta_0$. We train three SPFs: $\hat Y_k(x) = \hat\beta_k^\intercal x, k = 1,2,3$, where $\hat \beta_k$ is the OLS estimate of $\beta_k$ using $\tdk$. We consider $\cvcs$ and $\cvws$ for generalist predictions. 
	
	In Figure \ref{fig:sec3-main}(d-e), we plot $\wsU(w)$ and $\csU(w)$ as a function of $(w_1,w_2)$ when $w\in \Delta_2$. The contour lines are straight for $\wsU$ and curved for $\csU$. This illustrates that the geometries of the estimated utilities from $\cvws$ and $\cvcs$ are generally different. We compare $\cvws$ and $\cvcs$ in generalist predictions. Specifically, we simulate a dataset with size $n=10,000$, $\{(x_i, y_i);i=1,\ldots,n\}$ from $P_0$. We calculate the predicted values for all samples using different PFs. Denote the vectors of predicted values using the three SPFs by $\hat Y_k(X) = (\hat Y_k(x_i),i=1,\ldots,n)$, $k = 1,2,3$, predicted values of $\cvws$ by $\wsY(X) = (\wsY( x_i),i=1,\ldots,n)$, $\cvcs$ by $\csY(X) = (\csY(x_i),i=1,\ldots,n)$ and oracle by $Y_g(X) = (Y_g(x_i), i = 1,\ldots,n)$. We then project vectors $\hat Y_k(X)$, $k=1,2,3$, $\wsY(X)$, $\csY(X)$ and $Y_g(X)$ into $\mathbb R^2$ using principal component analysis (PCA), see Figure \ref{fig:sec3-main}(f). We note that $\wsY$ in this example deviates substantially from the oracle $Y_g$. On the other hand, $\csY$ is close to $Y_g$.
\end{example}

\begin{figure}[htbp]
	\centering
	\includegraphics[width=0.95\linewidth]{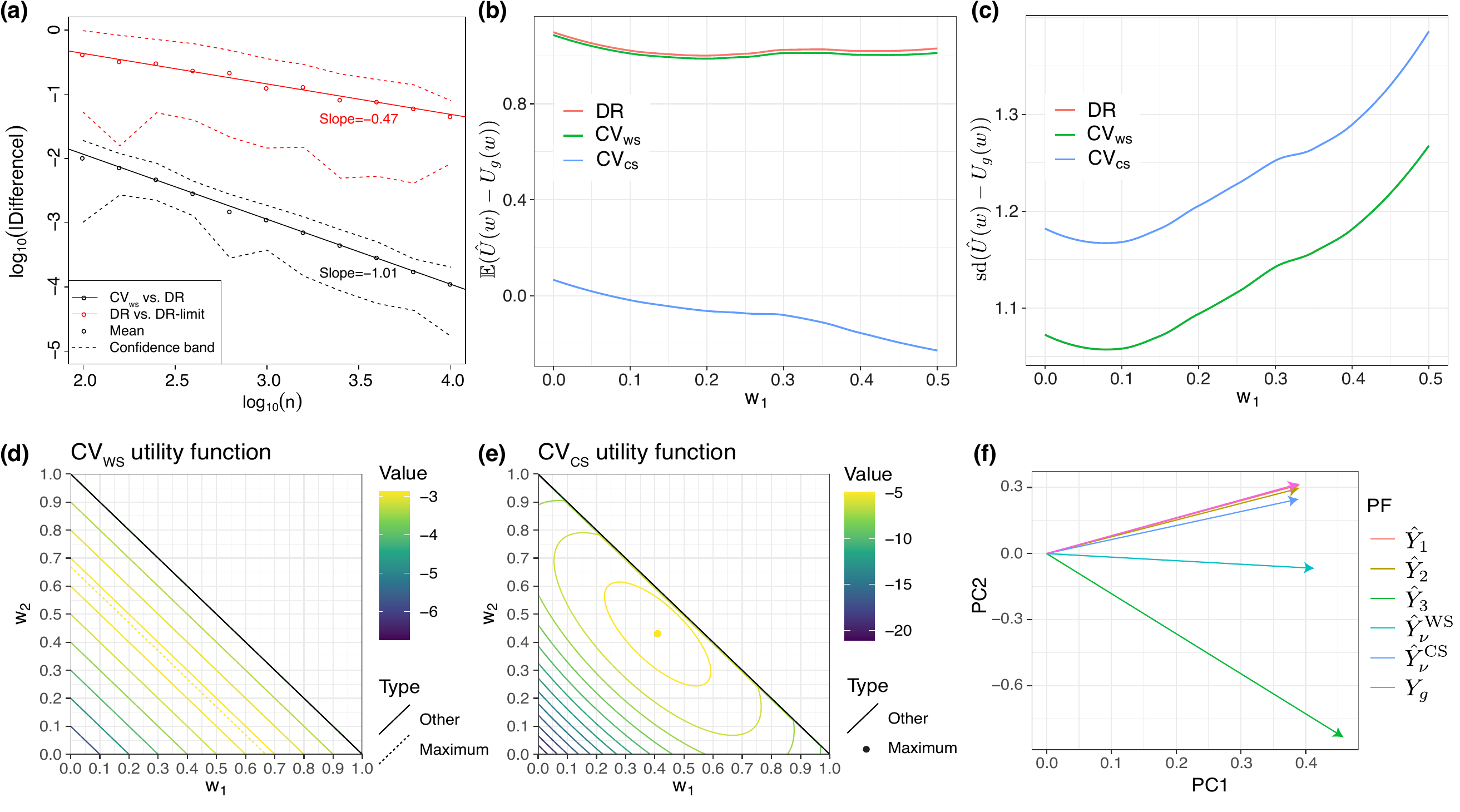}
	\caption{\small \it (\textbf{a}) Comparisons of DR and $\cvws$ in Example \ref{ex:two-y}. We illustrate $|\drU(w)-\wsU(w)|$ (black) and $|\drU(w) - \lim_{n\to\infty}\drU(w)|$ (red) at $w = \bone_K/K$ as a function of $n$. The dashed lines indicate the upper and lower fifth percentile of the differences simulation replicates. The solid lines illustrate the linear approximation of log-transformed average difference (black and red dots). The slopes approximate the rates of convergence of the differences when $n\to\infty$. (\textbf{b-c}) Comparisons of  bias and standard deviation of the utility estimates from DR, $\cvws$, $\cvcs$ stacking in Example \ref{ex:K-regression}. Note that in (\textbf{c}) the curves for DR and $\cvws$ stacking overlap with each other. (\textbf{d-f}) Comparisons of $\cvws$ and $\cvcs$ in generalist predictions (Example \ref{ex:three-study}). We visualize the contour plots of $\wsU(w)$ (\textbf{d}) and $\csU(w)$ (\textbf{e}). We use a dashed line and a dot to illustrate the maximizers of $\wsU(w)$ and $\csU(w)$ respectively. (\textbf{f}) PCA of different PFs. We include three SPFs $\hat Y_k$, two stacked PFs $\wsY$ and $\csY$, and the oracle PF $Y_g$.}
	\label{fig:sec3-main}
\end{figure}

\noindent \textbf{$\cvcs$ with arbitrary $D_t$'s.} $\cvcs$ can also be applied when $\mathcal D$ is not partitioned by study. In this case, we extend the definition of $\csU_k(w)$ to
$$
\csU_k(w) = \frac{1}{n_k} \sum_{i} u\left( \sum_{\ell, t}\mathbb I(k\notin s_t)w_{\ell,t}\hat Y_{t}^\ell(x_{i,k}), y_{i,k}\right),
$$
where $s_t = \{k:(i,k)\in D_t\text{ for some }i=1,\ldots,n_k\}$ is the set of study indices that are present in $D_t$. Similar to the partitioned by study setting, with exchangeable distributions $(P_1,P_2,\ldots)$, we have
$
\mathbb E\left[\csU_k(w)\mid \mathcal S_{-k}\right] = U_g(w_{-k}),
$
where $w_{-k}$ is now equal to $w$ except for all elements $w_{\ell,t}$ with $s_t$ containing $k$, which are set to zero. $\csU_k(w)$ are combined into $\csU(w)$. Using the same Taylor expansion argument as in the previous derivations,
$$
\sum_k \mathbb E\left[\csU_k(w)\bigg|\mathcal S_{-k}\right]/K \approx U_g(\Gamma w),
$$
where $\Gamma$ is a $KT\times KT$ diagonal matrix with the term corresponding to $w_{\ell,t}$ equal to $1-\#\{k\in s_t\}/K$. 

Next we assume that $1-\#\{k\in s_t\}/K>0$, that is, $D_t$ does not contain samples from all studies $k=1,\ldots,K$. Then the estimator of utility is
$$
\csU(w) = \sum_k\frac{1}{Kn_k}\sum_{i=1}^{n_k}u\left(\sum_{\ell,t}\frac{\mathbb I(k\notin s_t)}{1-\#\{k\in s_t\}/K}w_{\ell,t}\hat Y_t^\ell(x_{i,k}), y_{i,k}\right),
$$
which again has the approximate unbiasedness property as in the case where $\mathcal D$ is partitioned by study.

\section{Properties of stacking PFs for generalist predictions}
In consideration of the similarity of  $\drY_{w}$ and $\wsY_{w}$, we focus on the comparison between $\cvcs$ and DR stacking. We will assume that the data is generated from a hierarchical model, and that $\mathcal D$ is partitioned by study. We will only consider the scenario where studies differ in the conditional distribution $P(Y|X)$ while $P(X)$ remains the same across all studies.
%In the Discussion section, we explore and  discuss the results derived when this assumption is relaxed. \footnote{"In the Discussion section, we explore and  discuss the results derived when this assumption is relaxed."  CUT  unusual to call  Discussion in a method section}

We show that, with $\drw$ and $\csw$, the expected prediction error of the generalist PFs in future study $k>K$ is close
%\footnote{add "with large  sample sizes"  (not  trivial  as  k  could also  diverge)}
to the prediction error of an asymptotic oracle PF when $n_k$ and $K$ are large. The discrepancy between the prediction errors can be  bounded by a monotone function of $K$ and $\min_k n_k$.
%Second, we investigate in a special case the relative performance of DR stacking and $\cvcs$ in generalist predictions.

We consider arbitrary regression functions $\mathbb E(y_{i,k}|x_{i,k})$ with additive error: % and its distribution over different studies:
\begin{equation}
	\begin{aligned}
		y_{i,k} = f_k(x_{i,k})  + \epsilon_{i,k},~f_k \overset{iid}{\sim} F_f,~x_{i,k}\buildrel iid \over\sim F_X,
	\end{aligned}
	\label{eq:hier-mod}
\end{equation}
for $i=1,\ldots,n_k$ and $k=1,2,\ldots$. Here $f_k: \mathbb R^p\to\mathbb R$, $k\ge 1$, are random functions.
The mean of $F_f$ is indicated as $f_0$. The noise terms $\epsilon_{i,k}$ are independent with mean zero and variance $\sigma^2$.

Our bound is stated in Proposition~\ref{prop:bound} below, which requires the following assumptions:
\begin{enumerate}
	\item[A1.] There exists an $M_1<\infty$ such that a.e. (i.e., with probability 1) for any $k>0$ and $\ell\leq L$,
	$$
	\sup_{x\in\mathcal X} |f_k(x)|\leq M_1,~\sup_{x\in\mathcal X}|\hat Y_k^\ell(x)|\leq M_1,\text{ and,}~|y_{i,k}|\leq M_1.
	$$
	We require the first inequality to have probability 1 with respect to $F_f$. We require the second and third inequalities to have probability 1 with respect to the distribution of $\mathcal S$. These assumptions are not restrictive. For example, if $\mathcal X$ is a compact set and the outcomes $Y$ are bounded, then the SPFs trained with a penalized linear regression model, e.g. a LASSO regression, or with tree-based regression models satisfy the assumption.
	
	\item[A2.] There exist constants
	$M_2<\infty$, $p_{\ell}>0$ and random functions $Y_{k}^{\ell}$ for $k=1,\ldots, K$, $\ell = 1,\ldots,L$, such that $\sup_{x\in\mathcal X}|Y_{k}^\ell(x)|\leq M_1 ~a.e.,$ and for sufficiently large $n_k$, $k=1,\ldots,K$,
	$$
	\mathbb E\left[\int_{x} n_k^{2p_{\ell}}\left(\hat Y_{k}^\ell(x) - Y_{k}^\ell(x)\right)^2dF_X(x)\right] \leq M_2.
	$$
	Here the expectation is with respect to the distribution of $\mathcal S$ and $Y_k^\ell = \lim\limits_{n_k\to\infty}\hat Y_k^\ell$. For example, if $\|\mathbb E((X_k^{\intercal} X_k)^{-1})\|_F = O(1/n_k)$, where $\|\cdot\|_F$ is the Frobenius norm, and $\hat Y_k^\ell$ is an OLS regression function, then $Y_k^\ell(x) = (\lim \hat\beta_k)^\intercal x$ and $p_\ell < 1/2$ satisfy the above inequality.
\end{enumerate}

The asymptotic oracle PF is based on the limiting oracle generalist stacking weights $w_g^0$, which combine the limiting SPFs $Y_k^\ell$, $k = 1,\ldots,K, \ell = 1,\ldots, L$:
$$
w_g^0 = \argmax_{w\in W} \int_{x,y} u(Y_w(x),y)dP_0(x,y),
$$
where $Y_w = \sum_{\ell,t} w_{\ell,t} Y_t^\ell$ and $P_0$ (defined in Section 2.2) is the average joint distribution of $(X,Y)$ across studies $k\ge1$. The generalist prediction error of $Y_w$ is
$$
\psi(w) = \int_{x,y} (y - Y_w(x))^2dP_0(x,y) = w^\intercal \Sigma w - 2 b^\intercal w + \int_{y}y^2 dP_0(y).
$$
In Proposition \ref{prop:bound} we  compare $\drY_w$ and $\csY_w$ to the asymptotic oracle PF,  using  the  metrics $\mathbb E[\psi(\drw) - \psi(w_g^0)]$ and $\mathbb E[\psi(\csw) - \psi(w_g^0)]$. We show that both $\cvws$ and $\cvcs$ have the \textit{oracle property} in the sense that their generalist prediction error converge to that of the asymptotic oracle PF as the size of the studies $\min_k n_k$ and the number of studies $K$ both go to infinity.
\begin{proposition}
	\label{prop:bound}
	Let $K\geq 2$ and $w\in\Delta_{KL-1}$.  Consider $K$ training datasets and   all future studies $k=K+1,\ldots$ from model (\ref{eq:hier-mod}). If (A1) and (A2) hold, then
	\begin{gather*}
	\mathbb E\left(\psi(\drw) - \psi(w_g^0)\right) \leq C_0\sqrt{\log(KL)}K^{-1/2} + C_1(\min_kn_k)^{-\min_{\ell}p_{\ell}},\\
	\mathbb E\left(\psi(\hat w^{\text{CS}}) - \psi(w_g^0)\right) \leq C_0'\sqrt{\log(KL)}K^{-1/2} + C_1'(\min_kn_k)^{-\min_{\ell}p_{\ell}},
	\end{gather*}
	where the expectations are taken over the joint distribution of the data $\mathcal S$. $C_0$, $C_0'$, $C_1$ and $C_1'$ are constants, independent of $K$ and $n_k$.
\end{proposition}
%Note that by definition, $Y_{w_g^0}$ attains the smallest expected prediction error on future studies for any $\mathcal S$. It can be interpreted as the ``optimal'' stacking PF.
%The above proposition shows that if we have enough studies and samples in each study, then the estimated generalist PFs $\drY$ and $\csY$ have similar accuracy  compared to $Y_{w_g^0}$.

We then compare the performance of DR and $\cvcs$ stacking. Consider first the setting in Example \ref{ex:K-regression} (see Appendix for derivations).
\setcounter{example}{1}
\begin{example}[DR and $\cvcs$]
	We set $n_k = n$ for $k = 1,\ldots, K$. Consider first the bias (see Appendix for details):
	\begin{align*}
		&\mathbb E\left(\drU(w)-U_g(w)\right) =
		\frac{p(p+1)\sigma^2}{Kn(n-p-1)}\|w\|_2^2 + \frac{2(p\sigma_\beta^2 + p\sigma^2/n)}{K}w^\intercal \bone_K\\
		&\mathbb E\left(\csU(w),U_g(w)\right) = \frac{\|\beta_0\|_2^2\sum_{k\neq k'}w_kw_{k'}}{(K-1)^2}  - \frac{\|\beta_0\|_2^2 + p(\sigma_\beta^2+\frac{\sigma^2}{n-p-1})}{K-1}\|w\|_2^2.
	\end{align*}
	$\mathbb E(\drU(w)-U_g(w))\geq 0$ and $\mathbb E(\csU(w)-U_g(w))\leq 0$ for any $w\in\Delta_{K-1}$. As $\|\beta_0\|_2/\sigma_\beta \to 0$, the absolute bias of DR stacking becomes uniformly larger than $\cvcs$. On contrary, as $\|\beta_0\|_2/\sigma_\beta\to\infty$, the absolute bias of DR stacking becomes uniformly smaller than $\cvcs$.
	
	Although smaller absolute bias does not necessarily imply better generalist prediction accuracy, the relative prediction error of DR stacking to $\cvcs$, measured by $\mathbb E(\psi(\drw) - \psi(\csw)$, follows a similar trend as $\sigma_\beta$ changes. In Figure \ref{fig:zero-out-comp}, we use simulated data to illustrate $\mathbb E[\psi(\drw) - \psi(\csw)]$ as a function of $\sigma_\beta$. We generate data from the model in (\ref{eq:diff-sim}) and set $p = 10$, $\beta_0 = \bone_p$ and $n=200$. We examine two scenarios $K=2$ and $K=9$. $\mathbb E(\psi(w))$ is calculated with 1,000 replicates. As shown in the figures, when $\sigma_\beta$ is small, $\cvcs$ has lower generalist prediction accuracy than DR stacking ($\mathbb E(\psi(\drw) - \psi(\csw))<0$). When $\sigma_\beta$ is large, $\cvcs$ outperforms DR stacking.
	
	In Proposition \ref{prop:comp}, we derive the relationship between $\sigma_\beta$ and $\mathbb E(\psi(\drw) - \psi(\csw))$ when $n_k = n\to\infty$, $k=1,\ldots,K$, and $K\leq p/2$ is fixed.
\end{example}

\begin{proposition}
	\label{prop:comp}
	Consider the model in Example~\ref{ex:K-regression} with an OLS learner ($L=1$). Assume $p>4$, $K\leq p/2$ and $W = \mathbb R^K$. If $\beta_0 = \bzero$, then $\lim\limits_{n\to\infty}\mathbb E(\psi(\drw) - \psi(\csw))>0$ for every $\sigma_\beta>0$. If $\beta_0 \neq \bzero$, $$\lim\limits_{\sigma_\beta\to\infty}[\lim\limits_{n\to\infty}\mathbb E(\psi(\drw) - \psi(\csw))]\to\infty,$$ and $\lim\limits_{n\to\infty}\mathbb E(\psi(\drw) - \psi(\csw))\uparrow 0$ as $\sigma_\beta\downarrow 0$.
\end{proposition}

In Example \ref{ex:K-regression} with large sample sizes and relatively small number of studies, if the level of heterogeneity across studies $\sigma_\beta$ is sufficiently large, $\cvcs$ outperforms DR stacking. On the other hand, if $\sigma_\beta$ is close to zero, DR stacking tends to work better. Of note, if $\sigma_\beta$ is exceedingly high, neither approaches could work as cross-study predictions of a SPF would be problematic.

\begin{figure}[htbp]
	\centering
	\includegraphics[scale=0.32]{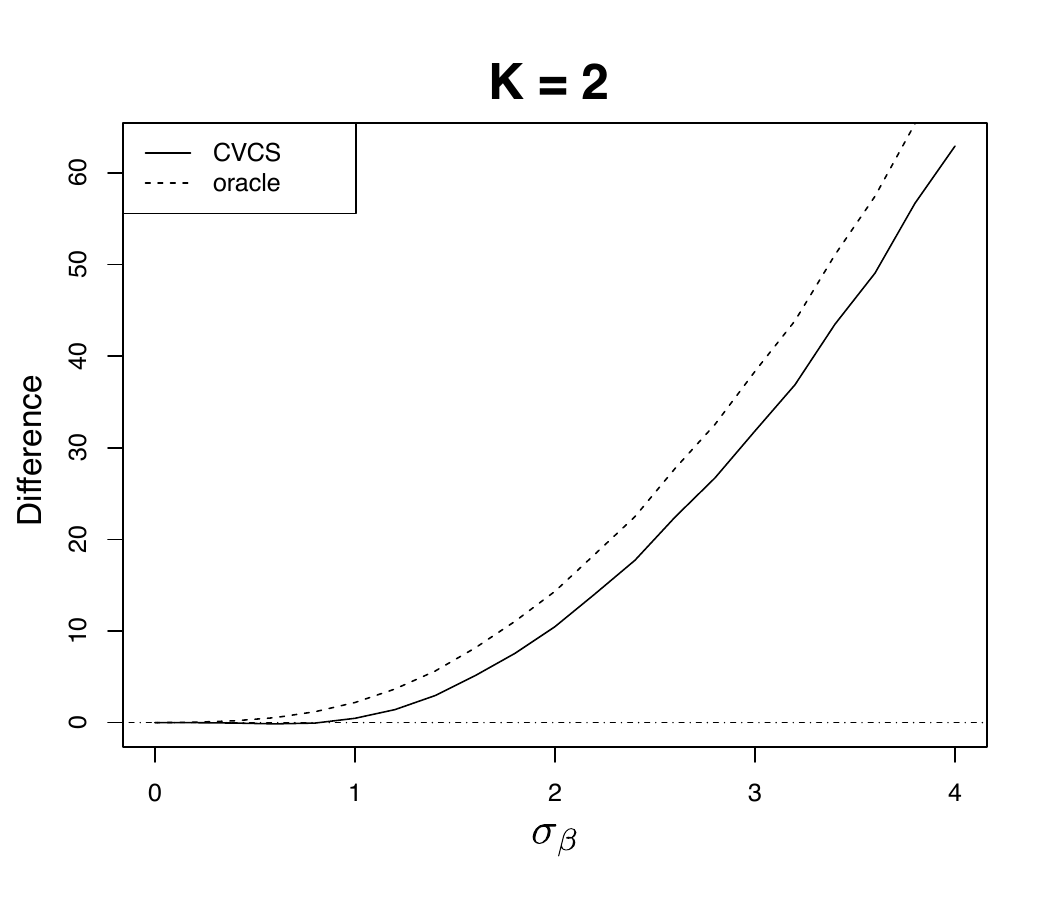}
	\includegraphics[scale=0.32]{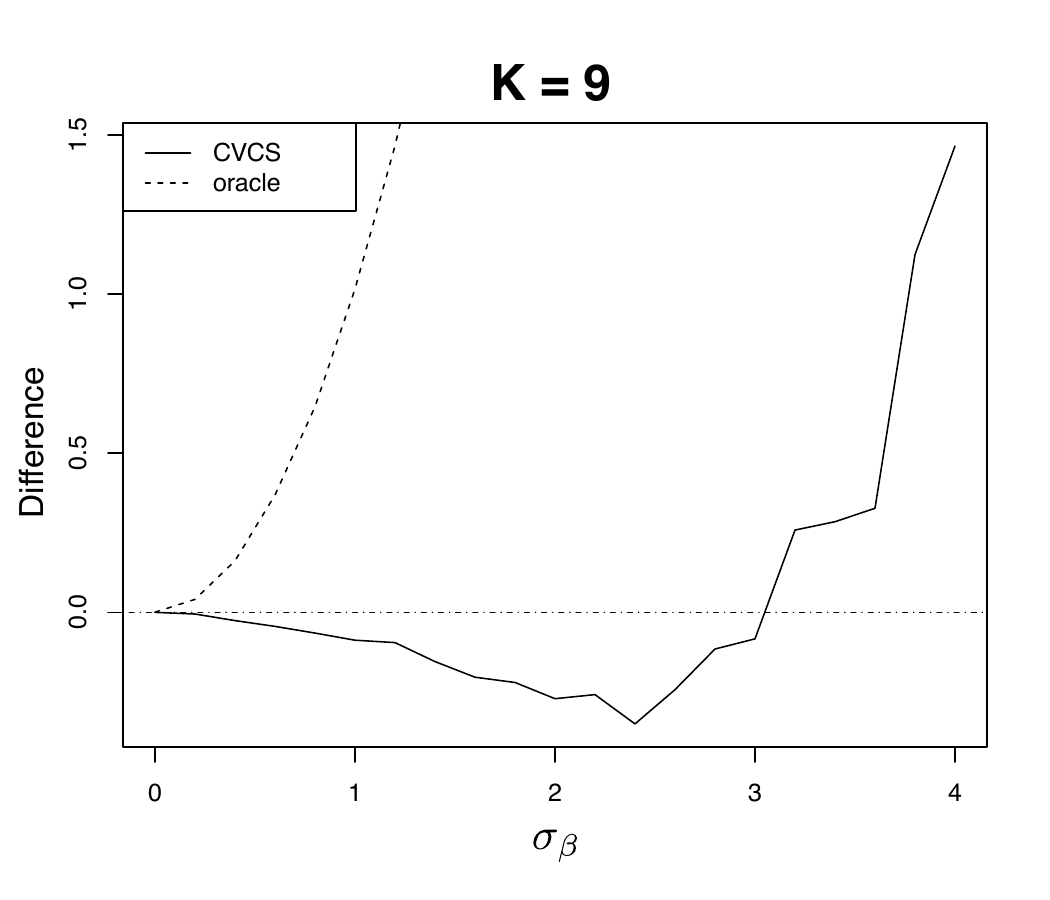}
	\caption{\small \it Comparison of DR stacking and $\cvcs$ when $K=2$ (left) and $K=9$ (right). The plots illustrate the differences $\mathbb E(\psi(\drw) - \psi(\csw))$ (CVCS) and $\mathbb E(\psi(\drw) - \psi(w_g^0))$ (oracle). We set $p=10$, $\beta_0 = \bone_K$, $n = 200$, and vary $\sigma_\beta$. The expected values are calculated with 1,000 replicates.}
	\label{fig:zero-out-comp}
\end{figure}

\section{Simulation studies}
\label{sec:simulation-studies}
We investigate empirically whether the error bounds in Proposition \ref{prop:bound} are tight, and how the difference $\mathbb E(\psi(\drw)) - \mathbb E(\psi(\csw))$ changes as $K$ and the inter-study heterogeneity $\sigma_\beta$ vary.

We illustrate the behavior of $\mathbb E\left(\psi(\drw)-\psi(w_g^0)\right)$ from Proposition \ref{prop:bound} with a numeric example and compare the actual difference to the analytic upper bound as $n_k$ and $K$ change. We use the simulation setting in Example \ref{ex:K-regression} with each component of $\beta_k$ uniformly distributed in $[0,1]$. Each component of $x_{i,k}$ is uniform in $[-1,1]$ and $\epsilon_{i,k}$ is uniform in $[-1,1]$. We set $n_k = n$ for all $k$ and approximate $\mathbb E\left(\psi(\drw) - \psi(w_g^0)\right)$ with Monte Carlo simulations, for $n = 100, 200, 400$ as $K$ increases from $20$ to $50$ (Fig. \ref{fig:sec5-res}(a)) and for $K=5, 15, 20$ as $n$ increases from $20$ to $100$ (Fig. \ref{fig:sec5-res}(b)). We use the constraint $w\in \Delta_{K-1}$. We simulate 1000 replicates. The difference $\mathbb E(\psi(\drw)-\psi(w_g^0))$ is approximately a linear function of $\sqrt{\log{K}/K}$ (Fig. \ref{fig:sec5-res}(a)) and $n^{-1/2}$ (Fig. \ref{fig:sec5-res}(b)). The results for $\mathbb E(\psi(\csw) - \psi(w_g^0))$ are similar (see Figure \ref{fig:cs-bound}).

We then perform a simulation analysis focused on the main results in Proposition \ref{prop:comp}. This analysis identifies regions where $\mathbb E(\psi(\drw) - \psi(\csw))>0$. We use the data generating model in Example \ref{ex:K-regression}. We fix $\beta_0 = \bone_{10}$, vary $\sigma_\beta$ between 0 and 4, and the number of studies $K$ between 3 and 50. We then approximate $\mathbb E(\psi(\drw)) - \mathbb E(\psi(\csw))$ with Monte Carlo simulations (Fig.~\ref{fig:sec5-res}(c)). In this figure the regions of $K$ and $\sigma_\beta$ in which $\mathbb E(\psi(\drw)) - \mathbb E(\psi(\csw))<0$ are in blue. We can see that when $\sigma_\beta$ is large, it is preferable to use $\cvcs$ for generalist predictions, while at smaller values of $\sigma_\beta$ the choice depends on the number of studies, and there remains a region where data reuse is preferable.

\begin{figure}[t]
	\centering
	\includegraphics[width=0.95\linewidth]{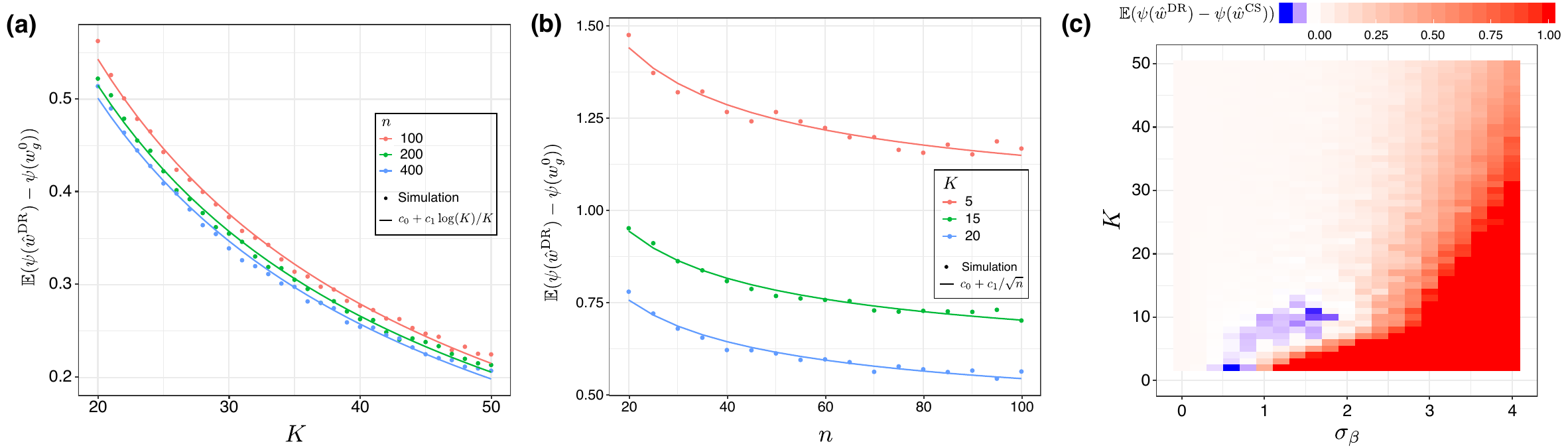}
	\caption{\small \it(\textbf{a-b}) $\mathbb E(\psi(\drw)-\psi(w_g^0))$ as a function of $K$ and $n$. Dots represent results from the Monte Carlo simulation. Lines illustrate the fitted functions $c_0 + c_1\log(K)/K$ (\textbf{a}) and $c_0 + c_1/\sqrt{n}$ (\textbf{b}) to the Monte Carlo results. (\textbf{c}) Comparison of generalist prediction accuracy of DR stacking and $\cvcs$ as measured by $\mathbb E(\psi(\drw) - \psi(\csw))$.}
	\label{fig:sec5-res}
\end{figure}

\section{Application}
As an illustration of the methods proposed in this article, we develop generalist PFs in an environmental health application. Our data consists of annual mortality rates for years~2000 to~2016 in 41,337 unique zip codes across the U.S. The annual average exposure to PM$_{2.5}$ is available for each zip code \citep{di2019ensemble}. In addition to the measurement of PM$_{2.5}$ and the annual mortality rate, we consider several zip code-level demographic variables, which include the percentage of ever smokers, percentage of the population who are black, median household income, median value of housing, percentage of the population below the poverty level, percentage of the population with education below high school, percentage of owner-occupied housing units, and population density.

Our goal is the generalist prediction of zip code-level annual mortality rates from 2000 to 2016 based on yearly exposure to PM$_{2.5}$ and zip-code-level demographic predictors in unseen studies. We also include year as a predictor to capture temporal trends of mortality rate. The zip codes are the sampling units, while studies are either States in one analysis or Counties in another. For generalist predictions at the state level, we randomly select 10 states to train an ensemble of state-specific PFs with random forest, and use DR stacking and $\cvcs$ to combine them and predict mortality in the remaining unseen 40 states. Here $k = 1, \ldots, 10$.  Each state contains on average 980 unique zip codes and 13,936 zip-code-by-year records. The metric we use to evaluate the accuracy of the stacked PF is the average RMSE across all years in the set-aside testing states or counties. We repeat this procedure 20 times. The county-level analysis is performed by restricting the analysis to the state of California and considering counties within it to be studied. For this county-level dataset, we combine 10 county-specific PFs and validate the stacked PFs using the remaining 47 counties. Each county contains on average 23 unique zip codes and 657 zip code by-year records. The results are shown in Figure~\ref{fig:sec6-res}.

\begin{figure}
	\centering
	\includegraphics[scale=0.7]{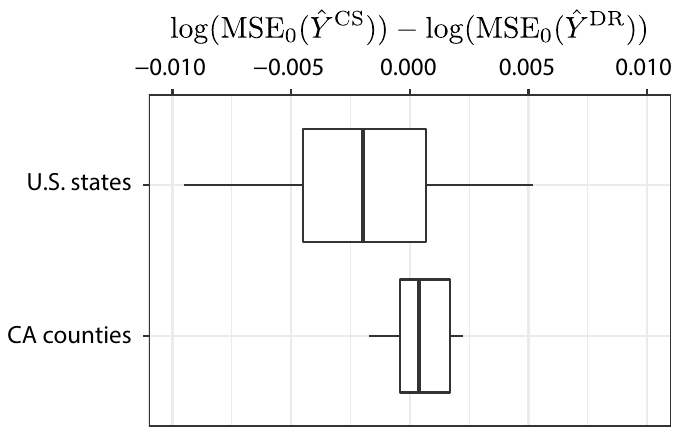}
	\caption{\small \it Comparison of DR stacking and $\cvcs$ in generalist predictions of mortality. Boxplots show the distributions of the differences in accuracy of the two stacking methods across 20 replicates, evaluated as average RMSE across all validation regions and years. In each replicate, we randomly select 10 regions to train, and evaluate the stacked PF on the remaining regions (39 test states in the U.S. and 47 test counties in California).}
	\label{fig:sec6-res}
\end{figure}

The State and County level analyses provide a useful contrast, as they reflect different degrees of heterogeneity in the underlying data distributions. For summary quantification, we estimate between-study heterogeneity by the variance of the predicted values from the 10 PFs on a single sample, with all predictors fixed at the national means. Across the 20 replicates, the mean of this heterogeneity metric is 0.017 for the state-specific PFs and 0.0029 for the county-specific PFs. 
In the state-level analysis, we thus combine study-specific PFs that reflect higher heterogeneity, and the performance of $\cvcs$ is slightly but consistently better compared to DR stacking (19 out of 20 replicates). In contrast, when we consider California and combine county-level models, with less heterogeneity, DR stacking has a smaller mean RMSE than $\cvcs$ (18 out of 20 replicates). This result is aligned with Proposition \ref{prop:comp}, which indicates for a fixed $K$, $\cvcs$ can outperform DR when the inter-study heterogeneity is large.

\section{Discussion}

We studied a multi-study ensemble learning framework based on stacking for generalist predictions. We compared CV and DR approaches for the selection of stacking weights of SPFs in the ensemble. We discussed oracle properties of the generalist PFs and provided a condition under which $\cvcs$ can outperform DR stacking. We illustrated these developments with simulations and applied our approach to predict mortality based on multi-state and multi-county datasets.

We illustrated in examples that DR stacking is nearly identical to $\cvws$ for generalist predictions, while $\cvcs$ is substantially different from both. We examined the differences $\mathbb E(\psi(\drw) - \psi(w_g^0))$ and $\mathbb E(\psi(\csw) - \psi(w_g^0))$ with Monte Carlo simulations and observed in simulations that they can be approximated by the bounds in Proposition \ref{prop:bound}. We also confirmed with numerical experiments and analytic results that $\cvcs$ outperforms DR stacking ($\mathbb E(\psi(\drw) - \psi(\csw))>0$) when the heterogeneity of $E(Y|X)$ across studies is high.

There are a couple of topics we would like to pursue further in future research. The procedure that we developed are applicable to any $\mathcal D$, beyond the partitioned-by-study case where $\mathcal D = \{\tdk;k=1,\ldots,K\}$, as indicated by Figure \ref{fig:D-example}. We would like to investigate the performance of our procedure for generic $\mathcal D$. We expect that a scenario where non-study-specific $\mathcal D$ is beneficial is when the multi-study dataset exhibits cluster structure with subsets of studies that share identical conditional distribution $P(Y|X)$. In this case, merging studies from the same cluster to form a training set $D$ increases the prediction accuracy of SPFs and in turn improves the performance of the resulting PFs. The other topic will target at non-exchangeable studies. As we discussed in Section \ref{sec:stacking}, study weights $\nu$ enable us to work with non-exchangeable studies, by specifying a set of $\nu_k$ that reflects the relationship between the individual studies, which might be informed by some external variables, such as time and geographic locations.

%%%%%%%%%%%%%%%%%%%%%%%%%%%%%%%%%%%%%%%%%%%%%%
%% Example with single Appendix:            %%
%%%%%%%%%%%%%%%%%%%%%%%%%%%%%%%%%%%%%%%%%%%%%%
%\begin{appendix}
%\section*{Title}\label{appn} %% if no title is needed, leave empty \section*{}.
%Appendices should be provided in \verb|{appendix}| environment,
%before Acknowledgements.
%
%If there is only one appendix,
%then please refer to it in text as \ldots\ in the \hyperref[appn]{Appendix}.
%\end{appendix}
%%%%%%%%%%%%%%%%%%%%%%%%%%%%%%%%%%%%%%%%%%%%%%%
%%% Example with multiple Appendixes:        %%
%%%%%%%%%%%%%%%%%%%%%%%%%%%%%%%%%%%%%%%%%%%%%%%
\begin{appendix}
\section{Proof of Proposition 1}
Let $M$ be such that $n$ is divisible by $M$. To prove $\sqrt{n}(\drU(w) - \lim_{n\to\infty} \drU(w))$ converges in distribution to $Z$, we note that $\drU$ is a continuous function of $\hat\beta_k, k = 1,\ldots,K$. The result follows from central limit theorem and delta method.

We now prove the other equality. Partition each study evenly into $M$ components and denote the $m$-th component of $X_k$ as $X_{k,m}$ and the responses as $Y_{k,m}$. Let $X_{k,-m}$ and $Y_{k,-m}$ denote the complements of $X_{k,m}$ and $Y_{k,m}$ respectively. The estimated regression coefficients are
$
\hat\beta_{k,m} = (X_{k,-m}^\intercal X_{k,-m})^{-1}X_{k,-m}^\intercal Y_{k,-m}.
$ Define $\delta \hat\beta_{k,m} = \hat\beta_{k,m} - \hat\beta_{k}$, where $\hat\beta_k = (X_k^\intercal X_k)^{-1}X_k^\intercal Y_k$ is the OLS estimate using all data in study $k$.

\begin{lemma}
	\label{lm:delta-beta}
	Under the assumption of Proposition 1, $\|\sum_{m}\delta\hat\beta_{k,m}\|_2 = O_p(1)$.
\end{lemma}
\begin{proof}
	First note the following relationship between $\hat \beta_{k,m}$ and $\hat\beta_{k}$ \citep{zhang1993model}:
	\begin{equation}
		\label{eq:prop1-link}
		\begin{aligned}
			\hat\beta_{k,m} = \hat \beta_k + (X_k^\intercal X_k)^{-1}X_{k,m}^\intercal(I_{n/M} - P_{k,m})^{-1}(X_{k,m}\hat \beta_k - Y_{k,m}),
		\end{aligned}
	\end{equation}
	where $P_{k,m} = X_{k,m}(X_{k}^\intercal X_k)^{-1}X_{k,m}^\intercal$ and $I_{n/M}$ is the identity matrix of size $n/M$. With central limit theorem, we have
	\begin{equation}
		\label{eq:prop1-rudi}
		\begin{gathered}
			\left\|\frac{X_k^\intercal X_k}{n} - I_p\right\|_F = O_p(\frac{1}{\sqrt{n}}),~
			\left\|\frac{X_{k,m}^\intercal X_{k,m}}{n} - \frac{1}{M}I_p\right\|_F = O_p(\frac{1}{\sqrt{n}}),~\|\delta\hat\beta_{k,m}\|_2 = O_p(\frac{1}{\sqrt{n}}),
		\end{gathered}
	\end{equation}
	where $\|\cdot\|_F$ is the Frobenius norm. From (\ref{eq:prop1-link}), it follows
	\begin{align*}
		\sum_m\delta\hat\beta_{k,m} = &\frac{(X_k^\intercal X_k)^{-1}}{M}\sum_mX_{k,m}^\intercal(I_{n/M} - P_{k,m})^{-1}(X_{k,m}\hat \beta_k - Y_{k,m})\\
		=&\frac{(X_k^\intercal X_k)^{-1}}{M}\sum_mX_{k,m}^\intercal(I_{n/M} + \sum_l P_{k,m}^l)(X_{k,m}\hat \beta_k - Y_{k,m})\\
		=&\frac{(X_k^\intercal X_k)^{-1}}{M}\sum_m\sum_l X_{k,m}^\intercal P_{k,m}^l(X_{k,m}\hat \beta_k - Y_{k,m}),
	\end{align*}
	where the second equation holds since the spectral norm (maximum eigenvalue) of $P_{k,m}$ is smaller than 1 when $n$ is large. Indeed, the spectral norm converges to $1/M$ as $n\to\infty$. The last equation holds since $(\sum_m X_{k,m}^\intercal X_{k,m})\hat\beta_k = X_k^\intercal Y_k = \sum_m X_{k,m}^\intercal Y_{k,m}$.
	
	We prove that $\|\sum_m X_{k,m}^\intercal P_{k,m}^l(X_{k,m}\hat \beta_k - Y_{k,m})\|_2 = O_p(1)$ for $l = 1,2, \ldots$, which will prove the lemma immediately. Indeed, since $Y_k = X_k\beta_k + \epsilon_k$, where $\epsilon_k = (\epsilon_{i,k};i=1,\ldots,n)$, we have
	\begin{align*}
		&\sum_m X_{k,m}^\intercal P_{k,m}^l(X_{k,m}\hat \beta_k - Y_{k,m}) \\
		= &(\sum_m X_{k,m}^\intercal P_{k,m}^lX_{k,m})(X_k^\intercal X_k)^{-1}X_k^\intercal \epsilon_k - \sum_m X_{k,m}^\intercal P_{k,m}^l\epsilon_{k,m},
	\end{align*}
	where $\epsilon_{k,m}$ is the $m$-th component of $\epsilon_k$. Using the assumption in Proposition 1, we know $\|X_{k,m}^\intercal\epsilon_{k,m}\|_2 = O_p(\sqrt{n})$ since $X_{k,m}^\intercal\epsilon_{k,m}$ is mean zero and its variance-covariance matrix is $n\sigma^2/MI_p$. From (\ref{eq:prop1-rudi}), we also have
	\begin{gather*}
		\|X_{k,m}^\intercal X_{k,m}(X_k^\intercal X_k)^{-1} - I_p/M\|_F = O_p(1/\sqrt{n})\\
		\|\sum_m X_{k,m}^\intercal P^l_{k,m} X_{k,m} - X_{k}^\intercal X_{k}/M^l\|_F = O_p(\sqrt{n}),\\
		\|\sum_m X_{k,m}^\intercal P_{k,m}^l \epsilon_{k,m} - X_k^\intercal \epsilon_k/M^l\|_2 = O_p(1).
	\end{gather*}
	It then follows that
	$$
	\left\|(\sum_m X_{k,m}^\intercal P^l_{k,m}X_{k,m})(X_k^\intercal X_k)^{-1}X_k^\intercal \epsilon_k - \sum_m X_{k,m}^\intercal P_{k,m}^l\epsilon_{k,m}\right\|_2 = O_p(1).
	$$
	The proof is completed by noting that $\|(X_k^\intercal X_k)^{-1}\|_F = O_p(1/n)$.
\end{proof}

We know
\begin{gather*}
	\drS_{i,i'} = (Kn)^{-1}\sum_{k=1}^K\hat \beta_{i}^\intercal X_{k}^\intercal X_{k} \hat\beta_{i'},~\drbb_i = (Kn)^{-1}\sum_{k=1}^K\hat \beta_i^\intercal X_{k}^\intercal Y_k,\\
	\wsS_{i,i'} = (Kn)^{-1}\sum_{k=1}^K\sum_{m=1}^M \hat \beta_{i,m}^\intercal X_{k,m}^\intercal X_{k,m} \hat\beta_{i',m},~\wsb_i = (Kn)^{-1}\sum_{k=1}^K\sum_{m=1}^M \hat \beta_{i,m}^\intercal X_{k,m}^\intercal Y_{k,m}.
\end{gather*}
It follows that
\begin{align*}
	n^{-1}\sum_m\hat\beta_{i,m}^\intercal X_{k,m}^\intercal X_{k,m} \hat\beta_{i', m} = &\hat\beta_i(n^{-1}X_{k}^\intercal X_{k})\hat\beta_{i'} + 2\sum_m \delta\hat\beta_{i,m}^\intercal(n^{-1}X_{k,m}^\intercal X_{k,m})\hat\beta_{i'} \\
	+& \sum_m \delta\hat\beta_{i,m}^\intercal(n^{-1}X_{k,m}^\intercal X_{k,m})\delta\hat\beta_{i',m}.
\end{align*}
(\ref{eq:prop1-rudi}) indicates that $\|\sum_m \delta\hat\beta_{i,m}^\intercal(n^{-1}X_{k,m}^\intercal X_{k,m})\delta\hat\beta_{i',m}\|_2 = O_p(1/n)$. And by invoking Lemma \ref{lm:delta-beta}, $\|\sum_m \delta\hat\beta_{i,m}^\intercal(n^{-1}X_{k,m}^\intercal X_{k,m})\hat\beta_{i'}\|_2 = \|\sum_m \delta\hat\beta_{i,m}^\intercal\hat\beta_{i'} + \sum_m \delta\hat\beta_{i,m}^\intercal(n^{-1}X_{k,m}^\intercal X_{k,m}-I_p/M)\hat\beta_{i'}\|_2=O_p(1/n)$. Therefore, $|\drS_{i,i'} - \wsS_{i,i'}| = O_p(1/n)$ for $i,i'=1,\ldots,K$. Similarly, we can prove $|\drbb_i - \wsb_i| = O_p(1/n)$ for $i = 1,\ldots,K$.

\section{Proof of Proposition~2}

Define $\hat\psi(w)$ as
$$
\drp(w) = \drU(w) - \frac{\sum_k n_k^{-1}Y_k^\intercal Y_k}{K} + \int_{y} y^2dP_0(y).
$$
Similarly, let $\csp(w) = \csU(w) - K^{-1}\sum_kn_k^{-1}Y_k^\intercal Y_k + \int_{y} y^2dP_0(y)$.

We first prove the following lemma on upper bounds of two differences $|\psi(\drw) - \psi(w_g^0)|$ and $|\psi(\csw) - \psi(w_g^0)|$.

\begin{lemma}
	$|\psi(\drw) - \psi(w_g^0)|$ and $\psi(\csw) - \psi(w_g^0)$ can be bounded as follows.
	\begin{gather*}
		|\psi(\drw) - \psi(w_g^0)| \leq 2\sup_{w\in W}|\psi(w) - \drp(w)|,\\
		|\psi(\csw) - \psi(w_g^0)| \leq 2\sup_{w\in W}|\psi(w) - \csp(w)|.
	\end{gather*}
	\label{lm:bound}
\end{lemma}

\begin{proof}
	We prove the inequality for $\drw$ and similar steps can be followed to verify the other inequality. Note that
	$$
	\psi(\drw) - \psi(w_g^0) = \psi(\drw) - \drp(\drw) + \drp(\drw) - \drp(w_g^0) +  \drp(w_g^0) - \psi(w_g^0).
	$$
	By definition $\psi(\drw) - \psi(w^0_g) \geq 0$ and $\drp (\drw) - \drp(w^0_g)\leq 0$, therefore
	$$
	|\psi(\drw) - \psi(w^0_g)| \leq |\psi(\drw) - \drp(\drw)| + |\drp(w^0_g) - \psi(w^0_g)| \leq 2\sup_{w\in W}|\psi(w) - \drp(w)|.
	$$
\end{proof}

When $W = \Delta_{K-1}$, we have
$
\sup_{w\in W}|\psi(w) - \drp(w)|  \leq \|\text{vec}(\Sigma-\drS)\|_{\infty} + \|b-\drbb\|_{\infty},
$
where $\|\cdot\|_{\infty}$ is the $L^{\infty}$-norm of a vector and $\text{vec}(\cdot)$ is the vectorization of a matrix. With Lemma \ref{lm:bound}, it follows
\begin{equation}
	\mathbb E[\psi(\drw) - \psi(w_g^0)] \leq 2\mathbb E\|\text{vec}(\Sigma-\drS)\|_{\infty} + 2\mathbb E\|b-\drbb\|_{\infty}.
\end{equation}
Similar results hold for $\mathbb E[\psi(\csw) - \psi(w_g^0)]$.
The following lemma provides an upper bound for $\mathbb E\|\text{vec}(\Sigma-\drS)\|$ and $\mathbb E\|\text{vec}(\Sigma-\csS)\|$.

\begin{lemma}
	If assumption A1 and A2 hold, we have the following bounds for $\mathbb E\|\text{vec}(\Sigma-\drS)\|_{\infty}$ and $\mathbb E\|\text{vec}(\Sigma-\csS)\|_{\infty}$.
	$$
	\begin{aligned}
		\mathbb E\|\text{vec}(\Sigma-\drS)\|_{\infty} \leq& 4\sqrt{2e}M_1^2\sqrt{\log(KL)/K} + 2M_1M_2(\min_k n_k)^{-\min_\ell p_\ell},\\
		\mathbb E\|\text{vec}(\Sigma-\csS)\|_{\infty} \leq& (4\sqrt{2e}+8)M_1^2\sqrt{\log(KL)/K} + 2M_1M_2(\min_k n_k)^{-\min_\ell p_\ell}.
	\end{aligned}
	$$
	\label{lm:sig-bound}
\end{lemma}

\begin{proof}
	First note that
	$$
	\drS_{k,\ell;k',\ell'} - \Sigma_{k,\ell;k',\ell'} = K^{-1}\sum_{s=1}^K n_s^{-1}\sum_{i=1}^{n_s} \left(\hat Y_{k}^{\ell}(x_{i,s})\hat Y_{k'}^{\ell'}(x_{i,s}) - \int _x Y_{k}^{\ell}(x)Y_{k'}^{\ell'}(x)dF_X(x)\right).
	$$
	Denote $\int _x Y_{k}^{\ell}(x)Y_{k'}^{\ell'}(x)dF_X(x)$ as $\langle Y_{k}^{\ell}, Y_{k'}^{\ell'}\rangle$, we have
	\begin{equation}
		\begin{aligned}
			\|\drS - \Sigma\|_{\infty} \leq & K^{-1}\sum_s n_s^{-1}\sum_i \left\|\left(\hat Y_k^\ell(x_{i,s})\left(\hat Y_{k'}^{\ell'}(x_{i,s}) - Y_{k'}^{\ell'}(x_{i,s})\right);k,k'\leq K, l,l'\leq L\right)\right\|_{\infty}\\
			+&K^{-1}\sum_s n_s^{-1}\sum_i \left\|\left(Y_{k'}^{\ell'}(x_{i,s})\left(\hat Y_k^\ell(x_{i,s})- Y_k^\ell(x_{i,s})\right);k,k'\leq K, l,l'\leq L\right)\right\|_{\infty}\\
			+&K^{-1}\left\|\left(\sum_s n_s^{-1}\sum_i \left(Y_k^\ell(x_{i,s})Y_{k'}^{\ell'}(x_{i,s}) - \langle Y_k^\ell, Y_{k'}^{\ell'}\rangle\right);k,k'\leq K, l,l'\leq L\right)\right\|_{\infty}
		\end{aligned}
		\label{eq:bound-inter1}
	\end{equation}
	By assumption A1, we have
	$$
	\left|\hat Y_k^\ell(x_{i,s})\left(\hat Y_{k'}^{\ell'}(x_{i,s}) - Y_{k'}^{\ell'}(x_{i,s})\right)\right|\leq M_1|\hat Y_{k'}^{\ell'}(x_{i,s}) - Y_{k'}^{\ell'}(x_{i,s})|.
	$$
	Combined with assumption $A2$, we have
	$$
	\mathbb E\left\|\left(\hat Y_k^\ell(x_{i,s})\left(\hat Y_{k'}^{\ell'}(x_{i,s}) - Y_{k'}^{\ell'}(x_{i,s})\right);k,k'\leq K, l,l'\leq L\right)\right\|_{\infty} \leq M_1M_2 (\min_k n_k)^{-\min_\ell p_\ell}.
	$$
	The same upper bound holds for the second term on the right-hand side of (\ref{eq:bound-inter1}).
	
	Define vector $\alpha_{i,s} = \left( Y_{k}^{\ell}(x_{i,s})Y_{k'}^{\ell'}(x_{i,s}) - \langle Y_{k}^{\ell}, Y_{k'}^{\ell'}\rangle;k,k'\leq K, l,l'\leq L\right)$. Based on Lemma 2.1 in \cite{juditsky2000functional}, we have
	\begin{align*}
		V\left(\sum_{s=1}^K n_s^{-1}\sum_{i=1}^{n_s} \alpha_{i,s}\right) \leq &V\left(\sum_{s=1}^{K-1} n_s^{-1}\sum_{i=1}^{n_s} \alpha_{i,s}\right) + (n_{K}^{-1}\sum_{i=1}^{n_K} \alpha_{i,K})^\intercal\nabla V\left(\sum_{s=1}^{K-1}n_s^{-1}\sum_i\alpha_{i,s}\right) \\
		+&c^*(M)\|n_K^{-1}\sum_i\alpha_{i,K}\|_{\infty}^2,
	\end{align*}
	where $M=K^2L^2$, $c^*(M) = 4e\log M$, $V(z) = 1/2\|z\|_q^2:\mathbb R^M\to\mathbb R$ and $q=2\log M$.
	It follows that
	\begin{equation}
		\mathbb E\left[V\left(\sum_{s=1}^K n_s^{-1}\sum_{i=1}^{n_s} \alpha_{i,s}\right)\right] \leq \mathbb E\left[V\left(\sum_{s=1}^{K-1} n_s^{-1}\sum_{i=1}^{n_s} \alpha_{i,s}\right)\right] + c^*(M)\mathbb E \|n_K^{-1}\sum_i\alpha_{i,K}\|_{\infty}^2,
		\label{eq:recur}
	\end{equation}
	since $\alpha_{i,s}$ and $\alpha_{i,s'}$ are independent conditioned on $f_1,\ldots,f_K$ when $s\neq s'$ and $\mathbb E(\alpha_{i,k})=0$. The inequality in (\ref{eq:recur}) implies a recursive relationship and repeatedly applying for $K$ times we get
	$$
	\mathbb E\left[V\left(\sum_{s=1}^K n_s^{-1}\sum_{i=1}^{n_s} \alpha_{i,s}\right)\right] \leq c^*(M)\sum_{s=1}^K n_s^{-2}\mathbb E\|\sum_{i}\alpha_{i,s}\|_{\infty}^2.
	$$
	
	By assumptions A1 and A2 again, we have
	$
	\left|Y_{k}^{\ell}(x_{i,s})Y_{k'}^{\ell'}(x_{i,s}) - \langle Y_{k}^{\ell},Y_{k'}^{\ell'} \rangle\right|\leq 2M_1^2,~a.e.
	$
	Therefore,
	$$
	\mathbb E\left[V\left(\sum_{s=1}^K n_s^{-1}\sum_{i=1}^{n_s} \alpha_{i,s}\right)\right] \leq c^*(M)4KM_1^4 = 32e\log(KL)KM_1^4.
	$$
	Since $V(z)\geq 1/2\|z\|_{\infty}^2$, it follows
	$$
	K^{-1}\mathbb E\|\sum_s n_s^{-1}\sum_i \alpha_{i,s}\|_{\infty}\leq K^{-1}\sqrt{32e\log(KL)KM_1^4} = 4\sqrt{2e}M_1^2\sqrt{\log(KL)/K}.
	$$
	
	The above steps also apply to prove of the bound of $\mathbb E\|\csS - \Sigma\|_\infty$ by noting that
	\begin{align*}
		|\csS_{k,\ell;k',\ell'} - \Sigma_{k,\ell;k',\ell'}| \leq& \left|\sum_s\sum_in_s^{-1}\left(\hat Y_k^\ell(x_{i,s})\hat Y_{k'}^{\ell'}(x_{i,s})- \langle Y_k^\ell, Y_{k'}^{\ell'}\rangle\right)\right| + \frac{4M_1^2}{K-1}\\
		\leq& \left|\sum_s\sum_in_s^{-1}\left(\hat Y_k^\ell(x_{i,s})\hat Y_{k'}^{\ell'}(x_{i,s})- \langle Y_k^\ell, Y_{k'}^{\ell'}\rangle\right)\right| + 8M_1^2\sqrt{\log(KL)/K}.
	\end{align*}
\end{proof}

\begin{lemma}
	If assumption A1 and A2 hold, then
	$$
	\begin{aligned}
		\mathbb E\|b - \drbb\|_{\infty} &\leq M_1M_2(\min_k n_k)^{-\min_\ell p_\ell} + (8\sqrt{2e} + 2) M_1^2\sqrt{\log(KL)/K},\\
		\mathbb E\|b - \csb\|_{\infty} &\leq 2M_1M_2(\min_k n_k)^{-\min_\ell p_\ell} + (16\sqrt{2e} + 6) M_1^2\sqrt{\log(KL)/K}.
	\end{aligned}
	$$
	\label{lm:b-bound}
\end{lemma}

\begin{proof}
	Note that
	\begin{align*}
		\|\drbb - b\|_\infty \leq &K^{-1}\sum_s n_s^{-1}\sum_i\left\| \left( \hat Y_{k}^{\ell}(x_{i,s})y_{i,s} - Y_{k}^{\ell}(x_{i,s})y_{i,s};k\leq K, \ell\leq L\right)\right\|_\infty +\\
		&K^{-1}\left\|\sum_s n_s^{-1}\sum_i\left( Y_{k}^{\ell}(x_{i,s})(y_{i,s} - f_s(x_{i,s}));k\leq K, \ell\leq L\right) \right\|_\infty+\\
		&K^{-1}\left\|\sum_s n_s^{-1}\sum_i\left(Y_{k}^{\ell}(x_{i,s})f_s(x_{i,s}) - \langle Y_k^\ell, f_0 \rangle;k\leq K, \ell\leq L\right)\right\|_\infty.
	\end{align*}
	
	With assumption A1 and A2, we have
	$$
	\mathbb E\left(K^{-1}\sum_s n_s^{-1}\sum_i\left\| \left( (\hat Y_{k}^{\ell}(x_{i,s}) - Y_{k}^{\ell}(x_{i,s}))y_{i,s};k\leq K, \ell\leq L\right)\right\|_\infty\right) \leq M_1M_2(\min_kn_k)^{-\min_\ell p_\ell}.
	$$
	Since $\epsilon_{i,s} = y_{i,s} - f_s(x_{i,s})$ is independent of $Y_k^\ell$ and $\mathbb E\epsilon_{i,s} = 0$, applying a similar recursive relationship as in (\ref{eq:recur}), we have
	\begin{align*}
		&\mathbb E\left(V\left(\sum_s n_s^{-1}\sum_i\left( Y_{k}^{\ell}(x_{i,s})(y_{i,s} - f_s(x_{i,s}));k\leq K, \ell\leq L\right)\right)\right) \\
		\leq& c^*(M)\sum_{s=1}^Kn_s^{-2}\mathbb E\|\sum_i (Y_k^\ell(x_{i,s})\epsilon_{i,s};k\leq K, \ell\leq L)\|_\infty^2\leq 4Kc^*(M)M_1^4.
	\end{align*}
	The last inequality holds due to Assumption A1. Therefore
	$$
	K^{-1}\mathbb E\left(\left\|\sum_s n_s^{-1}\sum_i\left( Y_{k}^{\ell}(x_{i,s})(y_{i,s} - f_s(x_{i,s}));k\leq K, \ell\leq L\right)\right\|_\infty\right)\leq 4\sqrt{2e}M_1^2\sqrt{\log(KL)/K}.
	$$
	Denote $\tilde\alpha_{i,s} = \left((Y_k^\ell(x_{i,s})f_s(x_{i,s}) - \langle Y_l^\ell, f_0\rangle)\mathbb I(s\neq k); k\leq K, \ell\leq L\right)$. It follows that
	\begin{align*}
		&\left\|\sum_s n_s^{-1}\sum_i\left(Y_{k}^{\ell}(x_{i,s})f_s(x_{i,s}) - \langle Y_k^\ell, f_0 \rangle;k\leq K, \ell\leq L\right)\right\|_{\infty} \\
		\leq & \left\|\sum_s n_s^{-1}\sum_i \tilde\alpha_{i,s}\right\|_\infty + \left\| \left(n_k^{-1}\sum_i(Y_k^\ell(x_{i,k})f_k(x_{i,k}) - \langle Y_k^\ell, f_0\rangle); k\leq K, \ell\leq L \right)\right\|_\infty\\
		\leq & \left\|\sum_s n_s^{-1}\sum_i \tilde\alpha_{i,s}\right\|_\infty + 2M_1^2.
	\end{align*}
	Note that by definition, $\mathbb E(\tilde\alpha_{i,s}) = 0$ and $\tilde\alpha_{i,s}$ is independent to $\tilde\alpha_{i,1},\ldots,\tilde\alpha_{i,s-1}$ conditioning on $f_1,\ldots f_{s-1}$. Applying (\ref{eq:recur}), we have $\left\|\sum_s n_s^{-1}\sum_i \tilde\alpha_{i,s}\right\|_\infty\leq 2\sqrt{Kc^*(M)}M_1^2$ and
	\begin{align*}
		&\left\|\sum_s n_s^{-1}\sum_i\left(Y_{k}^{\ell}(x_{i,s})f_s(x_{i,s}) - \langle Y_k^\ell, f_0 \rangle;k\leq K, \ell\leq L\right)\right\|_{\infty} \\
		\leq& 4\sqrt{2e}M_1^2\sqrt{\log(KL)/K} + 2M_1^2/K \leq (4\sqrt{2e}+2)M_1^2\sqrt{\log(KL)/K}.
	\end{align*}
	Therefore, we have
	$$
	\mathbb E\|\drbb - b\|_\infty \leq M_1M_2(\min_k n_k)^{-\min_\ell p_\ell} + (8\sqrt{2e}+2)M_1^2\sqrt{\log(KL)/K}.
	$$
	
	The proof is completed by noting that
	\begin{align*}
		\|\csb - b\|_\infty =& \left\|\frac{K}{K-1}(\drbb - b) + \frac{1}{K-1}\left(n_k^{-1}\sum_{i=1}^{n_k} \hat Y_k^\ell(x_{i,k})y_{i,k};k\leq K, \ell\leq L\right)\right\|_\infty\\
		\leq & \frac{K}{K-1}\|\drbb - b\|_\infty + \frac{M_1^2}{K-1}\leq 2\|\drbb - b\|_\infty + 2M_1^2\sqrt{\log(KL)/K}.
	\end{align*}
\end{proof}

Combining the results in Lemma \ref{lm:sig-bound} and Lemma \ref{lm:b-bound}, we have
\begin{align*}
	2\mathbb E\|\text{vec}(\Sigma - \drS)\|_\infty + 2 \mathbb E\|b - \drbb\|_\infty& \leq 6M_1M_2(\min_k n_k)^{-\min_\ell p_\ell} + (24\sqrt{2e}+4)M_1^2\sqrt{\log(KL)/K},\\
	2\mathbb E\|\text{vec}(\Sigma - \csS)\|_\infty + 2 \mathbb E\|b - \csb\|_\infty& \leq 8M_1M_2(\min_k n_k)^{-\min_\ell p_\ell} + (40\sqrt{2e}+28)M_1^2\sqrt{\log(KL)/K}.
\end{align*}
The proof is completed by applying Lemma \ref{lm:bound} on the above two inequalities.

\section{Proof of Proposition 3}
When $n\to \infty$, we have $\hat Y_k(x) \to \beta_k^\intercal x = Y_k(x)$. Therefore $\lim\limits_{n\to\infty}\psi(w) = (\bbeta w)^\intercal (\bbeta w) - 2(\bbeta w)^\intercal \beta_0 + (\|\beta_0\|_2^2 + p\sigma^2_\beta + 1)$, where $\bbeta = (\beta_1,\ldots,\beta_K)$, and
\begin{align*}
	& \drU(w) \to (\bbeta w)^\intercal(\bbeta w) - 2(\bbeta w)^\intercal \bar\beta + (\|\beta_0\|_2^2 + p\sigma_\beta^2 + 1),\\
	& \csU(w) \to \frac{K^2-2K}{(K-1)^2}(\bbeta w)^\intercal(\bbeta w) + \frac{K}{(K-1)^2}w^\intercal D w - \frac{2(\bbeta w)^\intercal \bbeta\bone_K - 2w^\intercal D \bone_K}{K-1} + (\|\beta_0\|_2^2 + p\sigma_\beta^2 + 1),
\end{align*}
where $\bar\bbeta = K^{-1}\sum_k \beta_k$ and $D = \text{diag}\{\|\beta_k\|_2^2,k\leq K\}$. If follows that $\drw = \bone_K/K$ and
$$
\csw = \frac{K-1}{K}\left((K-2)\bbeta^\intercal \bbeta + D\right)^{-1}(\bbeta^\intercal\bbeta - D)\bone_K.
$$
Since $K>2$, with Woodbury matrix identity
$$
\left((K-2)\bbeta^\intercal \bbeta + D\right)^{-1} = D^{-1} - D^{-1}\bbeta^\intercal\left(\frac{I_p}{K-2} + \bbeta D^{-1}\bbeta^\intercal\right)^{-1}\bbeta D^{-1},
$$
hence
$$
\csw = \frac{K-1}{K}\left(I_K - D^{-1}\bbeta^\intercal\left(\frac{I_p}{K-2} + \bbeta D^{-1}\bbeta^\intercal\right)^{-1}\bbeta\right)(D^{-1}\bbeta^\intercal\bbeta - I_K)\bone_K,
$$
and 
$$
\bbeta\csw = \frac{K-1}{K-2}\left(I_p - \frac{K-1}{K-2}\left(\bbeta D^{-1}\bbeta^\intercal + \frac{I_p}{K-2}\right)^{-1}\right)\bar\bbeta.
$$
We have
\begin{align*}
	\lim\limits_{n\to\infty}(\psi(\drw) - \psi(\csw)) = \|\bbeta\drw - \beta_0\|_2^2 -  \|\bbeta\csw - \beta_0\|_2^2.
\end{align*}

If $\beta_0 = \bm 0$, we only need to study the case where $\sigma_\beta = 1$ since $\mathbb E(\psi(\drw) - \psi(\csw)) = \sigma_\beta^2 \mathbb E(\psi(\drw_1) - \psi(\csw_1))$, where $\drw_1$ and $\csw_1$ are the stacking weights derived when $\sigma_\beta = 1$. Note in this case $\lim\limits_{n\to\infty}(\psi(\drw) - \psi(\csw)) = \|\bbeta \drw\|_2^2 - \|\bbeta \csw\|_2^2$. Denote $\tilde\beta_k = \beta_k/\|\beta_k\|_2$ and $\tilde\bbeta = (\tilde\beta_1,\ldots,\tilde\beta_K)$. Denote all non-zero $K$ eigenvalues and the corresponding eigenvectors of $\bbeta D^{-1}\bbeta = \tilde\bbeta\tilde\bbeta^\intercal$ as $\lambda_1,\ldots,\lambda_K$ and $v_1,\ldots,v_K$. It follows that
$$
\frac{K-1}{K-2}\left(I_p - \frac{K-1}{K-2}\left(\bbeta D^{-1}\bbeta^\intercal + \frac{I_p}{K-2}\right)^{-1}\right) = \sum_{k=1}^K \frac{(\lambda_k-1)(K-1)}{\lambda_k(K-2)+1}v_kv_k^\intercal,
$$
and
$$
\bbeta\csw = \sum_{k=1}^K \frac{(\lambda_k-1)(K-1)}{\lambda_k(K-2)+1}v_k\omega_k,
$$
where $\omega_k = v_k^\intercal \bar\bbeta$. It follows that
$$
\lim\limits_{n\to\infty}\left(\psi(\csw) - \psi(\drw)\right) = \sum_{k=1}^K \left(\left(\frac{(\lambda_k-1)(K-1)}{\lambda_k(K-2)+1}\right)^2-1\right)\omega_k^2.
$$
We note that $\tilde\beta_k$ follows a uniform distribution on a unit sphere. Based on the properties of spherical uniform distribution \citep{WatsonGeoffreyS1983Sos}, we can show that when $K\leq p/2$,
$$
\sum_{k=1}^K\mathbb E\left[\left(\left(\frac{(\lambda_k-1)(K-1)}{\lambda_k(K-2)+1}\right)^2-1\right)\omega_k^2\right]< 0
$$
Therefore when $\beta_0 = 0$, $\lim\limits_{n\to\infty}\mathbb E(\psi(\drw) - \psi(\csw))>0$ for any $K>2$.

If $\beta_0 \neq \bzero$, denote $\tilde\beta_0 = \beta_0/\|\beta_0\|_2$. When $\sigma_\beta = 0$,
$$
(\bbeta D^{-1}\bbeta^\intercal + I_p/(K-2))^{-1} = (K-2)I_p - \frac{K(K-2)^2}{(K-1)^2}\tilde\beta_0\tilde\beta_0^\intercal,
$$
and $\bbeta\csw = \beta_0 = \bbeta\drw$. Therefore $\lim\limits_{n\to\infty}\mathbb E(\psi(\drw) - \psi(\csw)) = 0$. When $\sigma_\beta = o(\|\beta_0\|_2)$, we have $\beta_k = \beta_0 + \eta_k$, where $\eta_k\sim N(\bzero, \sigma_\beta^2 I_p)$ and $\|\eta_k\|_2 = o_p(\|\beta_0\|_2)$. $\bbeta \drw = \beta_0 + \bar\eta$, where $\bar\eta = \sum_k\eta_k/K$. It follows that $\|\beta_k\|_2^{-1} \approx \|\beta_0\|_2^{-1}(1 - \beta_0^\intercal \eta_k/\|\beta_0\|_2^2)$. Thus
$
\tilde\beta_k = \beta_k/\|\beta_k\|_2 \approx \tilde\beta_0 + \|\beta_0\|_2^{-1}(I_p - \tilde\beta_0\tilde\beta_0^\intercal)\eta_k,
$ and 
$
\bbeta D^{-1}\bbeta^\intercal \approx K(\tilde\beta_0 + \tilde\eta)(\tilde\beta_0 + \tilde \eta)^\intercal,
$
where $\tilde\eta = \|\beta_0\|_2^{-1}(I_p - \tilde\beta_0\tilde\beta_0^\intercal) \bar\eta$. By noting that $\tilde\beta_0^\intercal \tilde\eta = 0$,
$$
(\bbeta D^{-1}\bbeta^\intercal + I_p/(K-2))^{-1} \approx (K-2)I_p - \frac{K(K-2)^2}{(K-1)^2}\left(\tilde\beta_0\tilde\beta_0^\intercal + \tilde\beta_0\tilde\eta^\intercal + \tilde\eta \tilde\beta_0^\intercal\right).
$$
Therefore, $\bbeta\csw \approx K(\tilde\beta_0\tilde\beta_0^\intercal + \tilde\beta_0\tilde\eta^\intercal + \tilde\eta\tilde\beta_0^\intercal)\bar\bbeta - (K-1)\bar\bbeta = \beta_0 + \bar\eta + \tilde\beta_0 \tilde\eta^\intercal \bar\eta + \tilde\eta\tilde\beta_0^\intercal \bar\eta$, and
\begin{align*}
	\lim\limits_{n\to\infty}\mathbb E(\psi(\drw) - \psi(\csw)) \approx& \mathbb E(\bar\eta^\intercal \bar\eta - (\bar\eta +  \tilde\beta_0 \tilde\eta^\intercal \bar\eta + \tilde\eta\tilde\beta_0^\intercal \bar\eta)^\intercal (\bar\eta +  \tilde\beta_0 \tilde\eta^\intercal \bar\eta + \tilde\eta\tilde\beta_0^\intercal \bar\eta))\\
	= & -4 \mathbb E(\bar\eta^\intercal \tilde\eta)(\tilde\beta_0^\intercal\bar\eta) - \mathbb E((\bar\eta^\intercal \tilde\eta)^2) - 3\mathbb E((\bar\eta^\intercal\tilde\eta)(\bar\eta^\intercal \tilde\beta_0)(\tilde\eta^\intercal \tilde\beta_0)) \\
	=& -\mathbb E((\bar\eta^\intercal \tilde\eta)^2).
\end{align*}
Since $\bar\eta\sim N(0,\sigma^2_\beta/KI_p)$, $|\lim\limits_{n\to\infty}\mathbb E(\psi(\drw) - \psi(\csw))| = O(\sigma_\beta^4)$. Therefore as $\sigma_\beta \downarrow 0$, $\mathbb E(\psi(\drw) - \psi(\csw)) \uparrow 0$.

When $\sigma_\beta\to \infty$, we note that $\beta_k' = \beta_k/\sigma_\beta$ converges in distribution to $N(0, I_p)$. Denote $\bbeta' = (\beta_1',\ldots,\beta_K')$. Based on the results when $\beta_0 = 0$, we know that $\lim\limits_{n\to\infty}\mathbb E(\|\bbeta'\drw\|_2^2 - \|\bbeta'\csw\|_2^2)$ converges to a positive constant as $\sigma_\beta\to\infty$. By noting that $\lim\limits_{n\to\infty}(\psi(\drw) - \psi(\csw)) = \sigma_\beta^2(\|\bbeta'\drw\|_2^2 - \|\bbeta'\csw\|_2^2)$, we have $\lim\limits_{n\to\infty}[\lim\limits_{\sigma_\beta\to\infty}\mathbb E(\psi(\drw) - \psi(\csw))] = \infty$.

\section{Derivations in Example 1 and 2}
\noindent\textbf{Example 1, Oracle}. We begin with the oracle case. Since $w\in\Delta_1$, the expected generalist utility is
$
U_g(w) = -(w_1\bar y_1 + w_2\bar y_2)^2 - 2 = -(w_1(\bar y_1 - \bar y_2) + \bar y_2)^2 - 2
$, $w_1 \in [0,1]$. If $\bar y_1 = \bar y_2$, any $w\in \Delta_1$ maximizes $U_g(w)$. If $\bar y_1 \neq \bar y_2$ and $\bar y_2/(\bar y_1 - \bar y_2)\in[0,1]$, $U_g(w)$ is maximized at $w_g = (\bar y_2/(\bar y_2 - \bar y_1), \bar y_1/(\bar y_1 - \bar y_2))$. Note $\bar y_2/(\bar y_2 - \bar y_1)\in[0,1]$ if and only if $\bar y_1\cdot \bar y_2 \leq 0$, in which case $\bar y_2/(\bar y_2 - \bar y_1) = |\bar y_2|/(|\bar y_1| + |\bar y_2|)$. If $\bar y_2/(\bar y_2 - \bar y_1)<0$, that is, $\bar y_1\cdot \bar y_2 >0$ and $|\bar y_1| > |\bar y_2|$, $w_g = (0,1)$. If $\bar y_1\cdot \bar y_2 >0$ and $|\bar y_1| < |\bar y_2|$, $\bar y_2/(\bar y_2 - \bar y_1)>1$ and $w_g = (1,0)$.\\

\noindent\textbf{Example 1, DR}. Moving to DR stacking, the estimated generalist utility is
$$
\drU(w) = -(w_1\bar y_1 + w_2\bar y_2)^2 + (\bar y_1 + \bar y_2)(w_1\bar y_1 + w_2\bar y_2) - \overline{y_1^2}/2 - \overline{y_2^2}/2,
$$
and it follows that $\drw = (1/2,1/2)$ and $\mathbb E(\drU(w) - U_g(w)) = \mathbb E((\bar y_1 + \bar y_2)(w_1\bar y_1 + w_2\bar y_2))$. We note that $\bar y_1$ and $\bar y_2$ are independent $N(0,\sigma^2 + 1/n)$. Therefore,
\begin{align*}
	\mathbb E\left(\drU(w) - U_g(w)\right) &= w_1\mathbb E(\bar y_1^2) + w_2 \mathbb E(\bar y_2^2) = \sigma^2 + 1/n,\\
	\mathbb E\left((\drU(w) - U_g(w))^2\right) &= w_1^2\mathbb E(\bar y_1^4) + w_2^2 \mathbb E(\bar y_2^4) + (1+2w_1w_2)\mathbb E(\bar y_1^2\bar y_2^2) \\
	=& (3w_1^2 + 3w_2^2 + 2w_1w_2 + 1)(\sigma^2 + 1/n)^2.
\end{align*}
Since $w\in\Delta_1$, we have $\var(\drU(w)-U_g(w)) = (2w_1^2 + 2w_2^2 + 1)(\sigma^2 + 1/n)^2$.\\

\noindent \textbf{Example 1, $\cvws$.} To show that $|\wsU(w) - \drU(w)| = O_p(1/n)$,
we fix $\mu = (\mu_1,\mu_2)^\intercal$ and it follows
\begin{align*}
	\wsS &= \drS + \frac{1}{(M-1)^2}\left[
	\begin{array}{cc}
		\frac{1}{M}\sum_m \bar y_{1,m}^2 - \bar y_1^2, & \frac{1}{M}\sum_m \bar y_{1,m}\bar y_{2,m} - \bar y_1\bar y_2 \\
		\frac{1}{M}\sum_m \bar y_{1,m}\bar y_{2,m} - \bar y_1\bar y_2, & \frac{1}{M}\sum_m \bar y_{2,m}^2 - \bar y_2^2 
	\end{array}
	\right],\text{ and }\\
	\wsb &= \drbb - \frac{1}{2(M-1)}\left[
	\begin{array}{c}
		\frac{1}{M}\sum_m \bar y_{1,m}^2 - \bar y_1^2 + \frac{1}{M}\sum_m \bar y_{1,m}\bar y_{2,m} - \bar y_1\bar y_2\\
		\frac{1}{M}\sum_m \bar y_{2,m}^2 - \bar y_2^2 + \frac{1}{M}\sum_m \bar y_{1,m}\bar y_{2,m} - \bar y_1\bar y_2
	\end{array}
	\right],
\end{align*}
where $\bar y_{k,m}$ is the average of $y_{i,k}$'s in the $m$-th fold. Note that $n(\sum_m \bar y_{k,m}^2/M - \bar y_k^2) \sim \chi^2_{M-1}$, $\mathbb E(\sum_m \bar y_{1,m}\bar y_{2,m}/M - \bar y_1\bar y_2) = 0$ and $\var(\sum_m \bar y_{1,m}\bar y_{2,m}/M - \bar y_1\bar y_2) = M^2/((M-1)n^2)$. It follows that $\sum_m \bar y_{k,m}^2/M - \bar y_k^2 = O_p(1/n)$ and $\sum_m \bar y_{1,m}\bar y_{2,m}/M - \bar y_1\bar y_2 = O_p(1/n)$.\\

\noindent \textbf{Example 1, $\cvcs$.} Note that
\begin{align*}
	\mathbb E\left( \csU(w) - U_g(w) \right) =& -w_1^2\mathbb E(\bar y_1^2) - w_2^2 \mathbb E(\bar y_2^2) = -(w_1^2 + w_2^2)(\sigma^2 + 1/n),\\
	\mathbb E\left( (\csU(w) - U_g(w))^2 \right) = & w_1^4\mathbb E(\bar y_1^4) + w_2^4 \mathbb E(\bar y_2^4) + 6w_1^2w_2^2E(\bar y_1^2\bar y_2^2) = 3(\sigma^2+1/n)^2(w_1^2 + w_2^2)^2.
\end{align*}
It follows that $\var(\csU(w;(1/2,1/2)) - U_g(w)) = 2(\sigma^2+1/n)^2(w_1^2+w_2^2)^2$. The $\text{MSE}_0$ for DR stacking and $\cvcs$ are
\begin{align*}
	\text{MSE}_0(\drw) = \frac{(\bar y_1 + \bar y_2)^2}{4} + 2,~~
	\text{MSE}_0(\csw) = \frac{(\bar y_1 + \bar y_2)^2\bar y_1^2\bar y_2^2}{(\bar y_1^2 + \bar y_2^2)^2} + 2,
\end{align*}
and
$$
\text{MSE}_0(\csw) - \text{MSE}_0(\drw) = \frac{(\bar y_1 + \bar y_2)^2(4\bar y_1^2\bar y_2^2 - (\bar y_1^2 + \bar y_2^2)^2)}{4(\bar y_1^2 + \bar y_2^2)^2} = - \frac{(\bar y_1 + \bar y_2)^2(\bar y_1^2 - \bar y_2^2)^2}{4(\bar y_1^2 + \bar y_2^2)^2}.
$$

\noindent \textbf{Example 2, $\cvcs$.} Let $X_k = (x_{1,k},\ldots,x_{n_k,k})^\intercal$, $Y_k = (y_{1,k},\ldots,y_{n_k,k})^\intercal$, and $\hat Y_{k'}^\ell(X_k) = (\hat Y^\ell_{k'}(x_{i,k});i\leq n_k)^\intercal$. We have
\begin{equation*}
	\begin{gathered}
		\drU(w) = w^\intercal \drS w - 2(\drbb )^\intercal w + K^{-1}\sum_{k}n^{-1}_k\sum_iy_{i,k}^2,\\
		\csU(w) = w^\intercal \hat \Sigma^{\text{CS}} w - 2(\hat b^{\text{CS}})^\intercal w + K^{-1}\sum_{k}n^{-1}_k\sum_iy_{i,k}^2.
	\end{gathered}
\end{equation*}
Recall that
\begin{gather*}
	\hat\Sigma^{\text{DR}}_{k,k';\ell,\ell'} = \sum_{i=1}^K \frac{\left(\hat Y_{k}^{\ell} (X_i)\right)^\intercal\hat Y_{k'}^{\ell'}(X_i)}{n_i K},~~\hat b^{\text{DR}}_{k;\ell} = \sum_{i=1}^K \frac{\left(\hat Y_{k}^{\ell} (X_i)\right)^\intercal Y_{i}}{n_iK},\\
	\hat\Sigma^{\text{CS}}_{k,k';\ell,\ell'} = \sum_{i\neq k,i\neq k'} \frac{K\left(\hat Y_{k}^{\ell} (X_i)\right)^\intercal\hat Y_{k'}^{\ell'}(X_i)}{n_i(K-1)^2},~~ \hat b_{k;\ell}^{\text{CS}} = \sum_{i\neq k} \frac{\left(\hat Y_{k}^{\ell} (X_i)\right)^\intercal Y_i}{n_i(K-1)}.
\end{gather*}
Here $\hat \Sigma^{\text{DR}}$ and $\hat \Sigma^{\text{CS}}$ are $KL\times KL$ matrices, and $w$, $\hat b^{\text{DR}}$ and $\hat b^{\text{CS}}$ are $KL$-dimensional vectors. $\hat \Sigma_{k,k';\ell,\ell'}$ is the element corresponding to $w_{k,\ell}$ and $w_{k',\ell'}$ while $\hat b_{k;\ell}$ corresponds to $w_{k,\ell}$. Since each row of $X_k$ follows $N(0, I_p)$, $X_k^\intercal X_k$ follows a Wishart distribution with $\mathbb E(X_k^\intercal X_k) = nI_p$ and $(X_k^\intercal X_k)^{-1}$ follows an inverse-Wishart distribution with $\mathbb E((X_k^\intercal X_k)^{-1}) = I_p/(n-p-1)$. Therefore,
\begin{align*}
	\mathbb E(\hat \beta_k^\intercal X_s^\intercal X_s \hat\beta_{k'}) =& \mathbb E(\beta_k^\intercal X_s^\intercal X_s \beta_{k'}) + \mathbb E(\epsilon_k^\intercal X_k(X_k^\intercal X_k)^{-1}X_s^\intercal X_s (X_{k'}^\intercal X_{k'})^{-1}X_{k'}^\intercal \epsilon_{k'})\\
	=& \left\{
	\begin{array}{cc}
		n\|\beta_0\|_2^2 + np\sigma_\beta^2 + p\sigma^2    & k = k' = s, \\
		n\|\beta_0\|_2^2 + np\sigma_\beta^2 + \frac{np\sigma^2}{n-p-1}    & k = k' \neq s,\\
		n\|\beta_0\|_2^2    & k\neq k',
	\end{array}
	\right.\\
	\mathbb E(\hat\beta_k^\intercal X_s^\intercal Y_s) =& \mathbb E(\beta_k^\intercal X_s^\intercal X_s\beta_s) + \mathbb E(\epsilon_k^\intercal X_k(X_k^\intercal X_k)^{-1}X_s^\intercal \epsilon_s) = n\|\beta_0\|_2^2 + \mathbb I(k=s)(np\sigma^2_\beta + p\sigma^2).
\end{align*}
In addition, $\mathbb E(\Sigma_{g;k,k'}) = \|\beta_0\|_2^2 + \mathbb I(k=k')(p\sigma_\beta^2 + p\sigma^2/(n-p-1))$ and $\mathbb E(b_k) = \|\beta_0\|_2^2$. It follows that $\mathbb E(\drS_{k,k'} - \Sigma_{g;k,k'}) = -\mathbb I(k = k')p(p+1)\sigma^2/(Kn(n-p-1))$. $\mathbb E(\csS_{k,k'} - \Sigma_{k,k'}) = (\|\beta_0\|_2^2 + p\sigma_\beta^2 + Kp\sigma^2/(n-p-1))/(K-1)$ if $k=k'$ and $\mathbb E(\csS_{k,k'} - \Sigma_{k,k'}) = \|\beta_0\|_2^2/(K-1)^2$ otherwise. $\mathbb E(\drbb_k - b_k) = p(\sigma^2_\beta + \sigma^2/n)/K$ and $\mathbb E(\csb_k - b_k) = 0$. As a result,
\begin{align*}
	\mathbb E\left(\drU(w) - U_g(w)\right) =& \frac{p(p+1)\sigma^2}{Kn(n-p-1)}\|w\|_2^2 + \frac{2(p\sigma^2_\beta + p\sigma^2/n)}{K}w^\intercal \bone_K,\\
	\mathbb E\left(\csU(w) - U_g(w)\right) =& \frac{\|\beta_0\|_2^2}{(K-1)^2}\sum_{k\neq k'} w_kw_{k'} - \frac{\|\beta_0\|_2^2 + p\sigma_\beta^2 + p\sigma^2/(n-p-1)}{K-1}\|w\|_2^2.
\end{align*}
$\mathbb E(\drU(w) - U_g(w))\geq 0$ for any $w\in \Delta_{K-1}$ and $\min(|\mathbb E(\drU(w) - U_g(w))|) = p\sigma^2(2+(p+1)/K/(n-p-1))/(Kn) + 2p\sigma_\beta^2/K$. $\mathbb E(\csU(w) - U_g(w))\leq 0$ for any $w\in \Delta_{K-1}$ and $\max(|\mathbb E(\csU(w) - U_g(w))|) = (\|\beta_0\|_2^2 + p\sigma_\beta^2 + p\sigma^2/(n-p-1))/(K-1)$. When $\|\beta_0\|_2 = o(\sigma_\beta)$ and $2p\leq n$, $|\mathbb E(\drU(w) - U_g(w))|\geq |\mathbb E(\csU(w) - U_g(w))|$ for any $w\in \Delta_{K-1}$. When $\sigma_\beta = o(\|\beta_0\|_2)$ and $\sigma = O(\sigma_\beta)$, $|\mathbb E(\drU(w)-U_g(w))|\leq |\mathbb E(\csU(w)-U_g(w))|$.

\section{Figure for $\mathbb E(\psi(\csw) - \psi(w_g^0))$}
\renewcommand\thefigure{\thesection.\arabic{figure}}
\setcounter{figure}{0}  
\begin{figure}[htbp]
	\centering
	\includegraphics[width=0.85\linewidth]{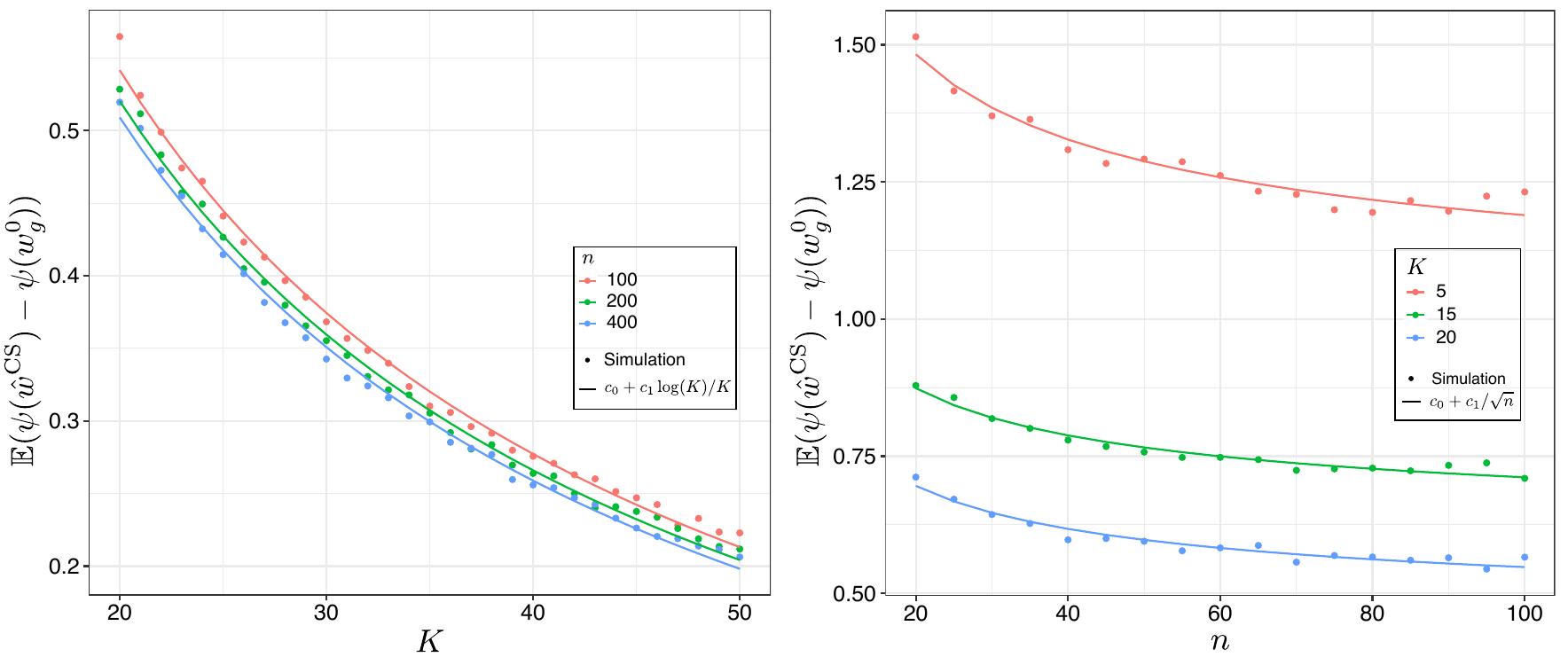}
	\caption{$\mathbb E(\psi(\csw)-\psi(w_g^0))$ as a function of $K$ and $n$. Dots represent results from the Monte Carlo simulation. Lines illustrate the fitted functions $c_0 + c_1\log(K)/K$ (\textbf{Left}) and $c_0 + c_1/\sqrt{n}$ (\textbf{Right}) to the Monte Carlo results. Simulation settings are described in Section 5.}
	\label{fig:cs-bound}
\end{figure}
\end{appendix}

%%%%%%%%%%%%%%%%%%%%%%%%%%%%%%%%%%%%%%%%%%%%%%
%% Acknowledgements                         %%
%% should be provided in the                %%
%% Acknowledgements section.                %%
%%%%%%%%%%%%%%%%%%%%%%%%%%%%%%%%%%%%%%%%%%%%%%
%\begin{acks}[Acknowledgments]
%The authors would like to thank the anonymous referees, an Associate
%Editor and the Editor for their constructive comments that improved the
%quality of this paper.
%\end{acks}

%%%%%%%%%%%%%%%%%%%%%%%%%%%%%%%%%%%%%%%%%%%%%%
%% Funding information, if any,             %%
%% should be provided in the                %%
%% funding section.                         %%
%%%%%%%%%%%%%%%%%%%%%%%%%%%%%%%%%%%%%%%%%%%%%%
\section*{Acknowledgements}
Research was supported by the U.S.A.'s National Science Foundation grant NSF-DMS1810829, the U.S.A.'s National Institutes of Health grant NCI-5P30CA006516-54, NIA-R01AG066670, NIDA-R33DA042847, NIMH-R01MH120400 and NIH-UL1TR002541 and Gunderson Legacy Fund 041662.

\bibliographystyle{chicago}
\bibliography{reference}

\end{document}